\newif\ifdraft \drafttrue \documentclass[a4paper,UKenglish]{false-lipics-v2021}
\hideLIPIcs
\showcopyright
\nolinenumbers 
\usepackage{hyperref,csquotes}
\usepackage[,sorting=anyt,minalphanames=2,maxalphanames=3,maxbibnames=100,giveninits=true,style=alphabetic]{biblatex}

\bibliography{bibliography}

\usepackage{soul}
\setstcolor{red}

\usepackage{imakeidx, xhfill}
\makeindex[title=Index \& Acronyms, options= -s indexstylev5]
\makeatletter
\indexsetup{level=\section, noclearpage, firstpagestyle=headings, headers={\thetitle}{Victor Marsault and Antoine Meyer}}
\makeatother

\newcommand{\semindex}[2]{\index{#2 (#1)}\index{#1|aka{#2}}}

\usepackage{amsmath, amsfonts, mathtools, amsthm, amssymb, stmaryrd, etoolbox, xargs}
\SetSymbolFont{stmry}{bold}{U}{stmry}{m}{n}
\usepackage{calc}
\usepackage{graphicx}
\usepackage{pdflscape}
\usepackage{afterpage}
\usepackage{caption}
\usepackage{booktabs, tabularx, colortbl}
\definecolor{tabularrowcolor}{gray}{0.9}
\definecolor{tabularrowtitle}{gray}{0.85}

\usepackage{etoolbox}

\newenvironment{subthm}{\begin{enumerate}}{\end{enumerate}}

\usepackage{tikz}
\usetikzlibrary{calc, arrows.meta, shapes}

\DeclareRobustCommand\sbseries{\fontseries{sb}\selectfont}
\newcommand{\defemph}[1]{\textcolor{blue!50!black}{\sbseries#1}}

\newcommand{\subwalkord}[1]{\ifstrequal{#1}{<}{\mathrel{\sqsubset}}{\ifstrequal{#1}{>}{\mathrel{\sqsupset}}{\ifdefequal{#1}{\leq}{\mathrel{\sqsubseteq}}{\ifdefequal{#1}{\geq}{\mathrel{\sqsupseteq}}{\subwalkord\leq#1}}}}}\newcommand{\factorord}[1]{\ifstrequal{#1}{<}{\mathrel{\prec}}{\ifstrequal{#1}{>}{\mathrel{\succ}}{\ifdefequal{#1}{\leq}{\mathrel{\preceq}}{\ifdefequal{#1}{\geq}{\mathrel{\succeq}}{\subwordrel\leq#1}}}}}

\newcommand{\subwordord}[1]{\ifstrequal{#1}{<}{\mathrel{\vartriangleleft}}{\ifstrequal{#1}{>}{\mathrel{\vartriangleright}}{\ifdefequal{#1}{\leq}{\mathrel{\trianglelefteq}}{\ifdefequal{#1}{\geq}{\mathrel{\trianglerighteq}}{\subwordord\leq#1}}}}}

\newcommand{\N}{\mathbb{N}}

\newcommand{\semfont}[1]{{\normalfont\textsf{#1}}}

\newcommand{\trailsem}{\semfont{Tr}}
\newcommand{\acyclicsem}{\semfont{Ac}}
\newcommand{\twoacyclicsem}{\semfont{2Ac}}
\newcommand{\simplesem}{\semfont{SWC}}
\newcommand{\shortesttrailsem}{\semfont{ShT}}
\newcommand{\shortestwalksem}{\semfont{Sh}}
\newcommand{\cheapestwalksem}{\semfont{ChW}}

\newcommand{\simplerunsem}{\semfont{SR}}

\newcommand{\bindingtrailsem}{\semfont{BT}}
\newcommand{\subwalkminsem}{\semfont{SM}}

\newcommand{\minmultsem}{\semfont{MM}}
\newcommand{\shortminsetsem}{\semfont{ShMS}}

\newcommand{\loglensem}{\semfont{LL}}
\newcommand{\vertexshortestcovering}{\semfont{ShVC}}
\newcommand{\edgeshortestcovering}{\semfont{ShEC}}
\newcommand{\atomshortestcoverer}{\semfont{ShAC}}
\newcommand{\givingupsem}{\semfont{GU}}

\newcommand{\shortlexsem}{\semfont{ShL}}

\newcommand{\gql}[1]{{\normalfont \textbf{\texttt{#1}}}}

\newcommandx{\matchspace}[2][1=R,2=D]{\textsc{Match}(#1,#2)}

\newcommand{\subw}[1]{\langle#1\rangle}

\newcommand{\walks}{\mathsf{Walks}}

\newcommand{\pow}[1]{2^{#1}}

\usepackage{chemarrow,adjustbox}
\newcommand{\rarrow}[2][]{\mathbin{{\let\rightarrow\chemarrow \adjustbox{set depth={\dimexpr\depth-.65ex\relax},set height={\dimexpr\height-.15ex\relax}}{$\xrightarrow[\adjustbox{raise={.75ex}{1.25ex}}{\hskip2pt$\scriptstyle #2$\hskip3pt}]{\hskip2pt#1\hskip3pt}$}}}}

\DeclareMathOperator{\len}{\mathsf{Len}}
\DeclareMathOperator{\src}{\mathsf{Src}}
\DeclareMathOperator{\tgt}{\mathsf{Tgt}}
\DeclareMathOperator{\ep}{\mathsf{EndP}}
\DeclareMathOperator{\lbl}{\mathsf{Lbl}}
\DeclareMathOperator{\matches}{\mathsf{Matches}}

\DeclareMathOperator{\elemset}{\mathsf{ElemSet}}
\DeclareMathOperator{\vertexset}{\mathsf{VertSet}}
\DeclareMathOperator{\edgeset}{\mathsf{EdgeSet}}
\DeclareMathOperator{\elembag}{\mathsf{ElemBag}}
\newcommand{\trim}[1]{{#1}_t}

\NewDocumentCommand{\intint}{d[] d(] d[) d()}{\IfNoValue{#1}{\llbracket#1\rrbracket}}

\makeatletter
\newcommand{\vm@date@separator}{\hspace*{0.15ex}\rule[0.4\vm@date@height]{1ex}{0.07\vm@date@height}\hspace*{0.15ex}}
\newcommand{\vmdatefont}[1]{#1}
\newcommand{\isotoday}{\vmdatefont{
    \newdimen\vm@date@height \setbox0=\hbox{0123456789}\vm@date@height=\ht0 \advance\vm@date@height by -\dp0
    \the\year\vm@date@separator\two@digits{\month}\vm@date@separator\two@digits{\day}}}

\newcommand{\Cc}{\mathcal{C}}
\newcommand{\Dc}{\mathcal{D}}

\newcommand{\vertices}{\mathbb{V}}
\newcommand{\edges}{\mathbb{E}}
\newcommand{\labels}{\mathbb{L}}

\newcommand{\finsubset}{\subset_{\text{fin}}}

\makeatletter
\newcommand{\set}[2][]{\{#1{\kern1\nulldelimiterspace}#2{#1\kern1\nulldelimiterspace\}}}
\newcommand{\setst}{\@ifstar{\autosetst}{\paramsetst}}
\newcommand{\autosetst}[2]{\left\{\,#1\,\middle|\,#2\,\right\}}
\newcommand{\paramsetst}[3][]{#1\{\,#2\mathbin{#1|}#3{\,#1\}}}

\newcommand{\bag}{\@ifstar{\autobag}{\parambag}}
\newcommand{\autobag}[1]{\llbrace*#1\rrbrace*}
\newcommand{\parambag}[2][]{\llbrace[#1]#2\rrbrace[#1]}
\newcommand{\bagst}{\@ifstar{\autobagst}{\parambagst}}
\newcommand{\autobagst}[2]{\llbrace*\,#1\,\middle|\,#2\,\rrbrace*}
\newcommand{\parambagst}[3][]{\llbrace[#1]\,#2\mathbin{#1|}#3\,\rrbrace[#1]}

\newcommand{\llbrace}{\@ifstar{\leftllbrace}{\paramllbrace}}
\newcommand{\paramllbrace}[1][]{{#1\{\hspace*{-.25em}#1\{}}
\newcommand{\leftllbrace}{\left\lbrace\kern-3\nulldelimiterspace\middle\lbrace}
\newcommand{\rrbrace}{\@ifstar{\rightrrbrace}{\paramrrbrace}}
\newcommand{\paramrrbrace}[1][]{{#1\}}\hspace*{-.25em}{#1\}}}
\newcommand{\rightrrbrace}{\middle\rbrace\kern-3\nulldelimiterspace\right\rbrace}
\makeatother

\newcommand{\np}{\textbf{NP}}
\newcommand{\conp}{\textbf{coNP}}
\newcommand{\polydelay}{\textbf{DelayP}}
\newcommand{\ptime}{\textbf{P}}
\newcommand{\nl}{\textbf{NL}}

\usepackage[framemethod=TikZ]{mdframed}

\usepackage{xpatch}
\makeatletter
\xpatchcmd{\endmdframed}
  {\aftergroup\endmdf@trivlist\color@endgroup}
  {\endmdf@trivlist\color@endgroup\@doendpe}
  {}{}
\makeatother

\newlength{\problemmargin}
\newlength{\problempadding}
\newlength{\problemsep}
\newlength{\problemtitlepadding}
\newcommand{\problemfont}[1]{\textsc{#1}}
\makeatletter
\newenvironment{problem}[1]{\begin{mdframed}[frametitle={\hspace*{\dimexpr\problemtitlepadding-\problempadding\relax}#1},
        frametitlefont=\large\sc,
        innerleftmargin=\problempadding,
        innerrightmargin=\problempadding,
        innerbottommargin=\problempadding,
        innertopmargin=\problempadding,
        frametitleaboveskip=\problemtitlepadding,
        frametitlebelowskip=\problemtitlepadding,
        usetwoside=false,
        leftmargin=\problemmargin,
        rightmargin=\problemmargin,
        skipabove={.75\baselineskip\@plus.2\baselineskip\@minus.2\baselineskip},
        skipbelow={.5\baselineskip\@plus.2\baselineskip\@minus.2\baselineskip},
frametitlerule=true,
        frametitlebackgroundcolor=black!10,
        nobreak,
        roundcorner=0pt,
        linewidth=0.5pt,
        frametitlerulewidth=0.5pt,
        ]\description }{\enddescription\end{mdframed}}
\makeatother

\newcommand{\CharR}[1]{P_{#1}}
\newcommand{\CharD}[1]{\Cc_{#1}}

\newlength{\CommentSignLength}
\usepackage[rightComments=false,beginComment={\hskip\CommentSignLength//~},beginLComment={//~},endLComment={}]{algpseudocodex}
\tikzset{algpxIndentLine/.style={dotted}}

\newcounter{lastalgoline}
\newenvironment{vmalgorithm}{\begin{algorithmic}[1]\setcounter{ALG@line}{\value{lastalgoline}}}
{\setcounter{lastalgoline}{\value{ALG@line}}\end{algorithmic}}

\algrenewcommand\algorithmicindent{1.25em}
\algrenewcommand\alglinenumber[1]{\makebox[1em][r]{\footnotesize #1.}}
\algrenewcommand\algorithmicdo{}
\algrenewcommand\algorithmicprocedure{\textbf{\texttt{type}}}
\algrenewcommand\algorithmicthen{}
\algrenewcommand\algorithmicforall{\textbf{for each}}

\newcommand{\True}{\textbf{True}}

 \newlength{\minnoderadius}
\newlength{\shortenlength}
\newlength{\nodelinewidth}
\newlength{\arrowwidth}
\def\loopangle{24}
\tikzset{bend angle=20,
    >={Stealth[width=6pt,length=6pt]},
    node/.style={circle, line width=\nodelinewidth, draw, black, inner sep=2pt, outer sep=.5*\nodelinewidth, minimum height=\minnoderadius, minimum width=\minnoderadius},
    vertex/.style={fill=black,inner sep=1.5pt,outer sep=.5pt, circle},
    state/.style={node},  
    run state/.style={draw,rounded rectangle},
    preedge/.style={-, draw, black, line width=1pt, rounded corners=5pt,pos=.4, shorten >=\shortenlength},
    edge/.style={preedge,->},
    label/.style={fill=white,font=\small},
    redge/.style={preedge,<-},
    initialedge/.style={edge, shorten > =0pt, shorten < =0pt},
    finaledge/.style={redge, shorten > =0pt, shorten < =0pt},
    borderedge/.style={edge, -, color=white, line width=5pt, shorten >=\shortenlength-2pt, >={Stealth[width=12pt,length=12pt]}},
    vm loop/.style={edge, pos=.5, looseness = 8},    
    graph loop/.style={edge, pos=.5, looseness = 40, shorten >=4pt},    
    road/.style={edge, color=fred},
    hl/.style={edge, line width = 1.5mm, color=borange, shorten >=1pt},
    ferry/.style={edge,color=fpurple},
    gas/.style={edge,color=fblue},
    start/.style={edge,color=fblue},
    end/.style={edge,color=fblue},
    >=stealth,
    north west loop/.style={vm loop, in={\the\numexpr 135 + \loopangle\relax}, 
                                     out ={\the\numexpr 135 - \loopangle\relax}},
    north east loop/.style={vm loop, in={\the\numexpr 45 + \loopangle\relax}, 
                                     out ={\the\numexpr 45 - \loopangle\relax}},
    south west loop/.style={vm loop, in={\the\numexpr -135 + \loopangle\relax}, 
                                     out ={\the\numexpr -135 - \loopangle\relax}},
    south east loop/.style={vm loop, in={\the\numexpr -45 + \loopangle\relax}, 
                                     out ={\the\numexpr -45 - \loopangle\relax}},
    north loop/.style={vm loop, in={\the\numexpr 90 + \loopangle\relax}, 
                                out ={\the\numexpr 90 - \loopangle\relax}},
    south loop/.style={vm loop, in={\the\numexpr 270 - \loopangle\relax}, 
                                out ={\the\numexpr 270 + \loopangle\relax}},
    east loop/.style={vm loop, in={\the\numexpr 0 + \loopangle\relax}, 
                                out ={\the\numexpr 0 - \loopangle\relax}},
    west loop/.style={vm loop, in={\the\numexpr 180 - \loopangle\relax}, 
                                out ={\the\numexpr 180 + \loopangle\relax}},                
    node distance = \nodedist,
}

\tikzset{
        state/.style={draw,rectangle, rounded corners},
        out state/.style={state, dashed},
        transition/.style={draw,-stealth,very thick, shorten >=2pt},
        out transition/.style={transition, dashed},
        p1/.style={rouge},
        p2/.style={bleu},
        p3/.style={vert},
        gadget/.style={draw},
        }

\newlength{\nodedist}
\newlength{\initfinaldist}
\newcommand{\initialfinal}[2][0]{\def\angleI{\the\numexpr #1 + 15 \relax}
    \def\angleII{\the\numexpr #1 - 15 \relax}
    \path (#2.\angleI) ++(#1:\initfinaldist) coordinate 
        (#2-initialfinal1-#1);
    \path[initialedge] (#2.\angleI) to         
        (#2-initialfinal1-#1);
    \path (#2.\angleII) ++(#1:\initfinaldist) coordinate     (#2-initialfinal2-#1);
    \path[finaledge] (#2.\angleII) to
      (#2-initialfinal2-#1);
}
\definecolor{vert}{rgb}{0,.55,0.20}
\definecolor{bleu}{rgb}{0,0,0.75}
\definecolor{rouge}{rgb}{.75,0,0}

\usepackage{multicol}
\newcommand{\thetitle}{Designing and Comparing RPQ Semantics}

\title{\thetitle}

\usepackage{adjustbox}

\author{Victor Marsault}
       {Univ Gustave Eiffel, CNRS, LIGM, F-77454 Marne-la-Vallée, France
        \and \url{https://victor.marsault.xyz}}
       {victor.marsault@univ-eiffel.fr}
       {https://orcid.org/0000-0002-2325-6004}
       {}
\author{Antoine Meyer}
       {Univ Gustave Eiffel, CNRS, LIGM, F-77454 Marne-la-Vallée, France}
       {antoine.meyer@univ-eiffel.fr}
       {https://orcid.org/0000-0003-4513-4347}
       {}

\authorrunning{V. Marsault and A. Meyer}

\relatedversion{A short version of this work is published in \cite{MarsaultMeyer2026}.} 

\Copyright{Victor Marsault and Antoine Meyer}

\ccsdesc{Theory of computation~Database query languages (principles)}
\ccsdesc{Information systems~Query languages for non-relational engines}
\ccsdesc{Information systems~Graph-based database models}
\ccsdesc{Theory of computation~Regular languages}

\keywords{
Regular Path Queries;
RPQ;
Semantics;
Graph Databases;
Pattern matching;
Regular Expression}

\date{\isotoday}

\acknowledgements{
This work stems from discussions with many colleagues, including Claire David, Amélie Gheerbrant, Steven Sailly and  Cristina Sirangelo.
In particular, we also thank Thomas Colcombet for the idea leading to the definition of the \vertexshortestcovering{} semantics, and Nadime Francis and Sara Khichane for introducing us to the flashlight search technique used in Proposition~\ref{p:extensibility+Membership=>Enumeration}.
Finally, we thank Yann Strozecki for his help with the enumeration complexity framework.
}

\begin{document}

\maketitle

\begin{abstract}
    Modern property graph database query languages such as Cypher, PGQL, GSQL, and the standard GQL draw inspiration from the formalism of regular path queries (RPQs).
    In order to output walks explicitly, they depart from the classical and well-studied \emph{homomorphism} semantics.
    However, it then becomes difficult to present results to users because RPQs may match infinitely many walks.
    The aforementioned languages use ad-hoc criteria to select a finite subset of those matches. 
    For instance, Cypher uses \emph{trail} semantics, discarding walks with repeated edges; 
    PGQL and GSQL use \emph{shortest walk} semantics, retaining only the walks of minimal length among all matched walks; and
    GQL allows users to choose from several semantics.
    Even though there is academic research on these semantics, it focuses almost exclusively on evaluation efficiency.

    In an attempt to better understand, choose and design RPQ semantics, we present a framework to categorize and compare them according to other criteria.
    We formalize several possible properties, pertaining to the study of RPQ semantics seen as mathematical functions mapping a database and a query to a \emph{finite} set of walks.
    We show that some properties are mutually exclusive, or cannot be met.
    We also give several new RPQ semantics as examples. Some of them may provide ideas for the design of new semantics for future graph database query languages.
\end{abstract}

\section{Introduction}
\label{s:intro}

In the past decades, there has been increasing interest and use for graph database management systems (DBMS). 
The typical task when exploiting such databases consists in navigating the graph.
To achieve this, most modern query languages contain features inspired by
Regular Path Queries (RPQs) \cite{CruzMendelzonWood1987}.
RPQs and their generalizations were thoroughly studied in academic literature. 
Essentially, an RPQ is a regular expression used to query an edge-labeled graph. 
A walk (that is, an alternating sequence of vertices and edges) matches an RPQ $R$ if its label (that is, the concatenation of its edge labels) conforms to $R$. 

In industry, RPQs were added to the standard language SPARQL \cite{SPARQL-PP} for RDF.
They were also the basis of the \gql{MATCH} clause in Cypher \cite{Cypher,FrancisEtAl2018-arxiv}, the language introduced with the successful property graph DBMS Neo4j.
RPQs were an inspiration for other query languages for graph DBMS \cite{GSQL,PGQL}.
Due in part to the success of Cypher and Neo4j, a standard query language for property graphs was designed \cite{GQL-ISO}.
Its navigational part is heavily influenced by RPQs, and was made compatible with SQL \cite{SQLPGQ-ISO}.

RPQs are usually evaluated under \emph{homomorphism semantics} \cite{AnglesEtAl2017}.
In this setting, the answer to a query~$R$ is the set of all pairs $(s,t)$ such that there exists a walk matching~$R$ from $s$ to~$t$.
In other words, the answer consists of the set of \emph{endpoints} of all matching walks.
However, endpoint semantics turns out to be inadequate in some practical scenarios: users sometimes need to have access to whole matched walks, and not only to their endpoints.

RPQs typically match infinitely many walks, hence if queries are to return whole walks then their presentation to the user in full is problematic.
To address this issue, many real-life query languages for property graphs choose to return only a finite subset of matching walks to the user, selected in various ways.
Cypher uses \emph{trail semantics}: a matching walk is only returned if it contains no edge repetition.
Other languages like GSQL or PGQL use \emph{shortest semantics}: a matching walk is returned only if there is no shorter match with the same endpoints.
Another well-known example is \emph{acyclic semantics}: a matching walk is returned only if it contains no vertex repetition.
GQL allows the user to choose one of these semantics, among others.
We follow this modern approach in the current article: what we call an RPQ semantics is any specification of a way to select, among a potentially infinite set of matching walks, a \emph{finite} subset of walks to return to the user.

Recent work on RPQs under various semantics focus on the evaluation process.
Computational problems related to evaluation are untractable under acyclic \cite{BaganBonifatiGroz2020} or trail \cite{MartensNiewerthPopp2023} semantics.
On the other hand, there exist efficient algorithms for evaluation under shortest semantics \cite{MartensTrautner2018,DavidFrancisMarsault2024}.
One of the challenges of query evaluation is the sheer number of results, and there has been work on better, more succinct representations of the set of results~\cite{MartensNiewerth+2023,MartensNiewerth+2024}.
This draws a very one-sided picture: shortest semantics always yields algorithms that are much more efficient. 
However, the most popular system as of today (Neo4j) uses trail semantics, and our hypothesis is that this can be ascribed to other desirable properties.

\subparagraph{Contributions.} To start this investigation, we propose a formal framework to define and categorize RPQ semantics.
We call semantics like trail or acyclic \emph{filter-based}, because they act as post-filters: each walk is individually considered good (e.g., with no repetition) or bad (e.g. with repetitions), independently of other matches.
We call semantics like shortest \emph{order-based}, because they select the best matched walks according to some global criterion. 
More precisely, such semantics are defined using (partial) ordering relations on walks, and return all walks that are minimal (among matches) with respect to this ordering. 

We then define many properties that RPQ semantics may or may not possess, and discuss why these properties may be desirable.
To state a few, we introduce properties related to monotonicity, continuity, composability using rational operators, and coverage.
Since we define RPQ semantics in the abstract, one may easily come up with unpractical or even nonsensical semantics.  
Hence, we propose a list of properties that any reasonable semantics should possess.
We also show several impossibility theorems, in the form of sets of mutually exclusive properties.
The goal of these results is to facilitate the design of new and interesting RPQ semantics.
We often illustrate our definitions and statements with examples of new RPQ semantics (which are compiled in the index, page~\pageref{s:index}).
Some of them appear to be reasonable candidates for possible practical use.

\subparagraph{Outline.}
Section~\ref{s:prelim} contains preliminary vocabulary and definitions, in particular the models we use for data (labeled graph) and queries (regular expressions). 
Then, Section~\ref{s:semantics} presents our framework: 
what RPQ semantics are, how to define one, and which ones we are interested in.
Section~\ref{s:properties} introduces several properties that RPQ semantics might possess.
We summarize the properties of several usual semantics, and provide a few general statement (e.g., which properties are mutually exclusive).
We also discuss the desirability of those properties.
Section~\ref{sec:computational-problems} recalls some known results about computational complexity.

In this document, we introduce many RPQ semantics, each with an acronym, and many properties for them. An index is provided on page~\pageref{s:index}.

\subparagraph{Related work.}
As far as we know, few RPQ semantics were studied in academic work other than the endpoint, acyclic, trail and shortest semantics.
The fragment of the RDF query language SPARQL which is related to RPQs is called \emph{property paths} \cite{SPARQL-PP}, and in the process leading to their definition, one of the candidate semantics indeed computed sets of walks. 
In order to only present a finite number of results, it was decided that in this semantics the same walk could not be iterated using a Kleene star.
This candidate semantics was discarded due to the fact that it was returning too many results \cite{ArenasConcaPerez2012}, a phenomenon sometimes referred to as the \emph{Yotabyte problem}.
It was finally decided that property paths in SPARQL would be evaluated under endpoint semantics, with a variation best explained in terms of rewriting:
property paths are automatically rewritten using other SPARQL constructs in order to obtain a query in which all property paths have a Kleene star as outermost operator (see Section~9.3, \emph{Property Paths and Equivalent Patterns}, of \cite{SPARQL-PP}).  This seemingly strange decision stems from the will to handle multiplicity consistently. 

In GQL\cite{GQL-ISO} and SQL/GPQ\cite{SQLPGQ-ISO}, the keyword \gql{SIMPLE} switches on a variant of acyclic semantics that returns not only simple walks but also simple cycles. 
It is usually not considered in theoretical articles, since it is computationally equivalent to acyclic semantics.

In~\cite{DavidFrancisMarsault2023}, the authors present two new semantics (simple run and binding trail) that filter walks based on the underlying matching process and not only on their topology.
Such semantics are outside the scope of this article (see~Sec.~\ref{s:out-of-scope-semantics}), but the introduction of~\cite{DavidFrancisMarsault2023} is particularly relevant to our work.
Indeed, their semantics are computationally better-behaved than trail, but worse than shortest.
In order to explain that shortest is not always the best semantics, they use informal criteria that are not based on theoretical algorithmic efficiency.
Our aim in the current text is to provide a more systematic argument.

\section{Preliminaries}
\label{s:prelim}

We assume the existence of $\vertices$, $\edges$, $\labels$, three disjoint countably infinite sets of \defemph{vertex identifiers}, \defemph{edge identifiers}
and \defemph{labels/letters}, respectively.
We will later define all databases and regular expressions over these sets.  This will allow us to define several notions independently of the underlying database.
We call \defemph{identifier} any element in~$\vertices\cup\edges$, \defemph{renaming} any permutation $\nu$ of~$\vertices\cup\edges$ such that $\nu(\vertices)=\vertices$ and $\nu(\edges)=\edges$.
We call \defemph{relabeling} any permutation of~$\labels$.

\subsection{Words, languages, expressions}
\label{s:formal languages}
An \defemph{alphabet} is a finite subset~$\Sigma$ of~$\labels$.
A \defemph{word} over~$\Sigma$ is a finite sequence of elements of~$\Sigma$;
the empty word is denoted by~$\varepsilon$.
The set of all words over~$\Sigma$ is denoted by~$\Sigma^*$.
Subsets of~$\Sigma^*$ are called \defemph{languages} over~$\Sigma$.
A word~$u = (a_1, \ldots, a_k)$ is usually written~$a_1\cdots a_k$, and $k$ is called the \defemph{length} of~$u$. The \defemph{concatenation} of two words $u = a_1\cdots a_k$ and $v = b_1\cdots b_\ell$ is~$u \cdot v = a_1 \cdots a_k b_1 \cdots b_\ell$, usually written simply $uv$. Concatenation is lifted to languages as follows:~$L_1\cdot L_2=\left\{w_1\cdot w_2 ~\middle|~w_1\in L_1,~w_2\in L_2\right\}$. 

A \defemph{regular expression}~$R$ over~$\Sigma$ is any formula over the set of letters, the unary function~${}^*$ and the binary functions $+$ and~$\cdot$ consistent with the following grammar.
\begin{equation} 
    R \mathrel{:{=}} \varepsilon \mid a \mid R^* \mid R\mathbin{\cdot}R \mid R+R   \quad\quad\text{with $a\in\Sigma$}
\end{equation}
Any such $R$ \defemph{denotes} a language $L(R)$ over~$\Sigma$, defined inductively as usual.
Any word in~$L(R)$ is said to \defemph{conform} to~$R$.

\subsection{Walks and databases}
\label{s:databases}

A \defemph{walk} is a finite nonempty alternating sequence of vertices in $\vertices$ and edges in $\edges$ that starts and ends with a vertex.
    We usually write such a walk~$w = (v_0,e_1,v_1,\ldots,e_k,v_k)$ as 
    $v_0 \rarrow{e_1} v_1 \rarrow{e_2} \cdots \rarrow{e_k} v_k$~.
    We call~$k$ the \defemph{length} of $w$ and denote it by $\len(w)$;
    we respectively call~$v_0$ and $v_k$ the \defemph{source} and \defemph{target} of~$w$ and denote them by~$\src(w)$ and $\tgt(w)$. The pair $(v_0, v_k)$ is the pair of \defemph{endpoints} of w, written~$\ep(w)$.
    Walks of length~$0$ are called \emph{trivial} and consist of a single vertex. 
    We denote by~$w[i]$ the edge~$e_i$, with~$i\in\{1,\ldots, k\}$ and by~$w\subw{i}$ the vertex~$v_i$, with~$i\in\{0,\ldots, k\}$.
    Note that walks have no particular topology.  Thus, it is possible for a walk~$w$ to be inconsistent with a database~$\Dc$ or even inconsistent in itself (see below).

    We let~$\walks$ denote the set of all walks (independently of any database).
    Two walks~$w$ and $w'$ \defemph{concatenate} if $\tgt(w) = \src(w')$,
    and their \defemph{concatenation}, written $w\cdot w'$, is the walk 
    \[
    w \cdot w' = \src(w) \rarrow{w[1]} \cdots \rarrow{w[k]} \tgt(w) \rarrow{w'[1]} \cdots \rarrow{w'[k']} \tgt(w')
    \]
    where $k = \len(w)$ and $k' = \len(w')$.  Note that $w \cdot w'$ is of length $k + k'$: it has $k + k'$ edges and $k + k' + 1$ vertices (vertex $\tgt(w) = \src(w')$ is not repeated). 
    We lift~$\cdot$ to sets of walks as usual: $W\cdot W'$ contains all walks~$w\cdot w'$ with~$w\in W$ and $w\in W'$ such that $w$ and $w'$ concatenate. 
    We lift the $\ast$ operator similarly:~$W^\ast$ contains all walks~$w_1\cdot w_2\cdots w_k$ where every $w_i$ belongs to $W$ and every pair $w_{i},w_{i+1}$ of consecutive walks concatenate.

\medskip

In this document, we model graph databases as directed, labeled, multi-edge graphs, and simply refer to them as \emph{databases}.
A \defemph{database} $\Dc$ is a tuple $(L_\Dc,V_\Dc, E_\Dc,\src_\Dc,\tgt_\Dc,\lbl_\Dc)$ where
    $L_\Dc \finsubset \labels$,
    $V_\Dc \finsubset \vertices$,
    $E_\Dc \finsubset \edges$
    and ${\src_\Dc}$, ${\tgt_\Dc}$ and ${\lbl_\Dc}$ are functions respectively mapping each edge in $E_\Dc$ to its source, target and label.

A walk~$w$ is \defemph{consistent with~$\Dc$} if all elements in $w$ are in $V_\Dc\cup E_\Dc$ and for every $0 < i \leq \len(w)$, $\src_\Dc(w[i]) = w\subw{i-1}$ and $\tgt_\Dc(w[i]) = w\subw{i}$.
The set of all walks consistent with~$\Dc$, or for short the set of \defemph{walks in $\Dc$}, is denoted by $\walks(\Dc)$.
We say that $n$ walks $w_1,\ldots,w_n$ are \defemph{consistent} if there is a database~$\Dc$ such that $\{w_1,\ldots,w_n\}\subseteq\walks(\Dc)$.
Whenever~$n=1$, we say $w_1$ is \defemph{self-consistent}.
We extend~$\lbl_\Dc$
to $\walks(\Dc)$ as expected: $\lbl_\Dc(w)$ is the concatenation of successive edge labels in $w$.
When the database~$\Dc$ is clear from the context, we abuse terminology and notation by using each edge~$e\in E_\Dc$ as the walk~$(\src_\Dc(e),e,\tgt_\Dc(e))$.
In particular we can ask whether a walk~$w$ and an edge~$e$ concatenate, and when they do, use the notation~$w\cdot e$.

Any renaming function~$\nu$ is extended to walks by applying~$\nu$ to all identifiers in it. 
It can also be extended to whole databases as follows: $\nu(\Dc)=\big( L_\Dc,\nu(V_\Dc), \nu(E_\Dc), \src_\Dc', \tgt_\Dc', \lbl_\Dc' \big)$, where $\src_{\Dc}'(x) = \nu(\src_{\Dc}(\nu^{-1}(x))$, $\tgt_{\Dc}' = \nu(\tgt_{\Dc}(\nu^{-1}(x))$ and $\lbl_{\Dc}'(x) = \lbl_{\Dc}(\nu^{-1}(x))$.
Similarly, any relabeling function~$\lambda \colon \labels \to \labels$ is extended to regular expressions by applying $\lambda$ to all atoms, and to databases as follows: $\lambda(\Dc) = \big(\lambda(L_\Dc),V_\Dc, E_\Dc, \src_\Dc, \tgt_\Dc, (\lambda \circ \lbl_\Dc) \big)$.

\section{Regular Path Queries (RPQ) and their semantics}
\label{s:semantics}

The concept of Regular Path Query (RPQ) consists in querying the content of a graph database~$\Dc$ thanks to regular expressions over its edge labels.

\begin{definition}[Matches]\label{d:matches}
    \index{Match}
    Given a regular expression~$R$ and a database~$\Dc$, we denote by $\matches(\Dc,R)$ the set of \defemph{matches of~$R$ in~$\Dc$}, which is the set of all walks in~$\Dc$ whose label conforms to~$R$: 
$\matches(\Dc,R) = \setst{w\in\walks(\Dc)}{\lbl_\Dc(w) \in L(R)}$. 
Given~$v,v'\in V_\Dc$, we will also use the notation~$\matches(\Dc,R,v,v')$ for the set of the walks~$w$ in~$\matches(\Dc,R)$ such that $\ep(w)=(v,v')$.
\end{definition}

\subsection{Arbitrary semantics}

\begin{definition}\label{d:arbitrary-semantics}
    \index{Arbitrary semantics}
    A \defemph{semantics}~$S$ is a function mapping a database~$\Dc$ and a query~$R$ to a finite subset $S(\Dc,R)$ of $\matches(\Dc,R)$.
\end{definition}

Aside from homomorphism semantics, which does not fit our framework because it does not return walks, the two most popular RPQ semantics are shortest-walk semantics and trail semantics.
We define them below.

\begin{example}[Trail semantics~$\trailsem$]
    \semindex{\trailsem}{Trail semantics}
    \label{ex:trail}
    Trail semantics is the semantics of Cypher~\cite{FrancisEtAl2018-sigmod} and may be used in GQL~\cite{DeutschEtAl2022,GQL-ISO} thanks to the keyword \gql{TRAIL}.
    It is defined as follows.   
    A walk $w$ is a \defemph{trail} if it has no duplicate edge, i.e.\@ for all $i \neq j$, $w[i] \neq w[j]$.
    Trail semantics is defined by:
$\trailsem(\Dc,R) = \setst*{w\in \matches(\Dc,R)}{w\text{ is a trail}}$.
\end{example}

\begin{example}[Shortest-walk semantics~$\shortestwalksem$]
    \semindex{\shortestwalksem}{Shortest semantics}
    \label{ex:shortest}    
    Shortest semantics is the original semantics of GSQL~\cite{GSQL} and PGQL~\cite{PGQL}; it may be used in GQL~\cite{DeutschEtAl2022,GQL-ISO} thanks to the keyword \gql{SHORTEST}.
    Let~$\Dc$ be a database and~$R$ a query.
    For every $s,t\in V_\Dc$, 
   we write $\ell_{(s,t)}$ for the length of the shortest walk in~$\matches(\Dc,R,s,t)$.
Then, shortest semantics is defined by: $\shortestwalksem(\Dc,R) = \bigcup_{s,t\in V_\Dc}
        \setst*{w\in \matches(\Dc,R,s,t)}{\len(w)\mathbin=\ell_{(s,t)}}$.
\end{example}

We focus on semantics that are \emph{walk-based}, as defined below. 
Intuitively, such semantics act as post-filters on matches: they only have access to the set of matched walks itself, not to the database, to the regular expression, or to the reason or way a given walk was matched by the query.

\begin{definition}
    \label{d:walk-based}
    \index{Walk-based semantics}
    A semantics~$S$ is \defemph{walk-based} if there is a function~$r_S \colon \pow{\walks}\to\pow{\walks}$ such that~$S(D,R)=r_S(\matches(D,R))$.
\end{definition}

Even though the previous definition seems very general, walk-based semantics still have to output walks which match the query due to Definition~\ref{d:arbitrary-semantics}. 
Instances of walk-based semantics are~$\trailsem$ and $\shortestwalksem$ given above, as well as shortest-trail semantics defined below.
On the other hand, we will give in Section~\ref{s:out-of-scope-semantics} a few examples of semantics that are not walk-based, such as binding-trail semantics and cheapest-walk semantics.

\begin{example}[Shortest-trail semantics $\shortesttrailsem$]\label{ex:shortest-trail}
    \semindex{\shortesttrailsem}{Shortest-trail semantics}
    Shortest-trail semantics may be used in GQL~\cite{DeutschEtAl2022,GQL-ISO} thanks to the keyphrase \gql{SHORTEST} \gql{TRAIL}.
    For every~$W\subseteq 2^\walks$, and every vertices~$s,t\in\vertices$,
    we let~$\ell_{(s,t)}(W)$ denote the length of the shortest trail~$w$ in~$W$ such that~$\ep(w)=(s,t)$, if it exists.
    Then, shortest-trail semantics is defined by: $\shortesttrailsem(\Dc,R) = r_\shortesttrailsem(\matches(D, R))$, where
$r_\shortesttrailsem(M)$ contains all trails~$w$ in~$M$ such that~$\len(w)=\ell_{\ep(w)}(M)$.
\end{example}

For every $\Dc$ and $R$, $\shortesttrailsem(\Dc,R)$ is clearly a subset of $\trailsem(\Dc,R)$, since it only returns trails matching $R$ in $\Dc$. However, it may not be a subset of $\shortestwalksem(\Dc, R)$ since the shortest walk matching $R$ between two vertices of $\Dc$ may or may not be a trail.

\subsection{Filter-based semantics}
\label{s:fiter-based}

\begin{definition}\label{d:filter-based}\index{Filter-based}
    A semantics~$S$ is called \defemph{filter-based} if there exists a function~$f_S \colon \walks\to\{\top,\bot\}$ such that 
    $S(\Dc,R)=\setst{w\in\matches(\Dc,R)}{f_S(w)=\top}$.
\end{definition}

Note that every filter-based semantics is also walk-based.
Trail semantics (Ex.~\ref{ex:trail}) is the paramount example of a filter-based semantics.
It is defined by the function $f_{\trailsem}$, 
that maps a walk~$w$ to~$\bot$ if there are $i<j$ such that $w[i]\mathbin=w[j]$, or to~$\top$ otherwise.
However, it is quite easy to see that neither shortest-walk semantics (Ex.\,\ref{ex:shortest}) nor shortest-trail semantics (Ex.~\ref{ex:shortest-trail}) are filter-based.
We give below three examples of filter-based semantics.

\begin{example}[Acyclic semantics~$\acyclicsem$]
    \semindex{\acyclicsem}{Acyclic semantics}
    Acyclic semantics (corresponding to the~\gql{ACYCLIC} keyword in~GQL) is the semantics defined by the filter $f_{\acyclicsem}$ that maps a walk~$w$ to $\bot$ if there are $i<j$ such that $w\subw{i}\mathbin=w\subw{j}$, and to~$\top$ otherwise.
\end{example}

\begin{example}[Simple-cycle-or-walk semantics~$\simplesem$]
    \label{ex:simple}
    \semindex{\simplesem}{Simple-walk-or-cycle semantics}
    Simple-cycle-or-walk semantics is a variant of $\acyclicsem$ that corresponds to keyword~\gql{SIMPLE} in GQL.
    The difference is that simple cycles are also returned. It is defined by the filter $f_{\simplesem}$ that maps a walk~$w$ to~$\bot$ if there are $i<j$ such that$(i,j)\neq(0,\len(w))$ and $w\subw{i}= w\subw{j}$, and to~$\top$ otherwise.
\end{example}

\begin{example}[Two-acyclic-sem~$\twoacyclicsem$]\label{ex:2-ac}
    \semindex{\twoacyclicsem}{Two-acyclic semantics}
    Two-acyclic semantics is a variant of $\acyclicsem$  where each vertex is allowed to appear once or twice, but not more.
    It is defined by the filter below.
\begin{equation*}
    f_{\twoacyclicsem} \colon w \mapsto \begin{cases}
        \bot & \text{if there are }i\mathbin< j\mathbin<k\text{ such that } w\subw{i} \mathbin= w\subw{j}\mathbin= w\subw{k}\\
        \top & \text{otherwise}
    \end{cases}
\end{equation*}
\end{example}

\subsection{Order-based semantics}

First, let us recall a few definitions about partial orders.
A \defemph{partial order} is a binary, transitive, antisymmetric and reflexive relation.
For any partial order~$\preceq$ we will denote by~$\prec$ the associated \defemph{strict partial order}.
A \defemph{well-partial-order} on a set~$X$, or \defemph{wpo}, is a partial order~$\preceq$ such that in every infinite sequence~$x_1, x_2, \ldots, x_k, \ldots $ of elements from~$X$ there are two positions~$i<j$ such that $x_i\preceq x_j$.
In other words, there exists no infinite strictly decreasing sequence nor any infinite sequence in which all elements are incomparable w.r.t.~$\preceq$. 
It is well-known that any subset of $X$ admits finitely many minimal elements w.r.t.\@ a wpo over~$X$, and that the restriction of any wpo over~$X$ to a subset of $X$ is also a wpo.

\begin{definition}[Order-based semantics]
    \label{d:order-based}\index{Order-based}\index{Suitable}
    Let~$\preceq$ be a partial order on~$\walks$ such that for every database~$\Dc$, $\preceq$ is a wpo on~$\walks(\Dc)$.
    Such an order is said to be \defemph{suitable}.
    The \defemph{semantics based on~$\preceq$}, denoted by~$S_\preceq$ is defined as follows for every database~$\Dc$ and regular expression~$R$ : 
$S_\preceq(\Dc,R) = \bigcup_{s,t\in V_\Dc} \min_\preceq(\matches(\Dc,R,s,t))$

\end{definition}

Note that a suitable order~$\preceq$ is a wpo on~$\matches(\Dc,R,s,t)$, since $\matches(\Dc,R,s,t)$ is a subset of~$\walks(\Dc)$.
It follows that~$S_\preceq(\Dc,R)$ is always finite. 
Note also that the ordering of two walks~$w,w'$ w.r.t.~$\preceq$ is irrelevant whenever $w$ and~$w'$ have different endpoints. 
The same is true for walks~$w$ and~$w'$ that are inconsistent.
It is sometimes useful to trim such irrelevant pairs from the orders we consider.

\begin{definition}\label{l:trimmed-order}
    \index{Trimmed}
    For every suitable order~$\preceq$ we let $\trim{\preceq}$ denote
    the order defined by: ${w \trim{\preceq} w'}$ iff $w$ and~$w'$ are consistent walks with the same endpoints such that $w \mathbin{\preceq} w'$.
    We say that $\preceq$ is \defemph{trimmed} if it is equal to $\trim{\preceq}$
\end{definition}

\begin{lemma}\label{l:order-irrelevant-different-ep}
    Each order-based semantics~$S_\preceq$ is identical to $S_{\trim{\preceq}}$.
\end{lemma}

A simple example of order-based semantics is the shortest-walk semantics~$\shortestwalksem$ (Ex.\,\ref{ex:shortest}). 
It is based on the partial order $\preceq_{\shortestwalksem}$, the reflexive closure of the strict partial order defined by $w \prec_{\shortestwalksem} w'$ if $\len(w) < \len(w')$.
It is \emph{not} a wpo on $\walks$: since~$\edges$ is infinite, there are infinitely many walks of a given length, and they are all incomparable w.r.t.\@ $\prec_{\shortestwalksem}$. 
However, one may verify that it is indeed a wpo on $\walks(\Dc)$ for every database~$\Dc$, and that it is therefore suitable.
Let us now define another order-based semantics.

\begin{definition}[Subwalk order]\label{d:subwalk-order}
    \semindex{\subwalkminsem}{Subwalk-minimal semantics}
    Let~$F$ be the following factor-inserting binary relation: $w \mathrel{F} w'$ if there are three walks~$w_1,w_2,w_3$ such that $w=w_1\cdot w_3$ and~$w'={w_1\cdot w_2 \cdot w_3}$.
    We define the \defemph{subwalk order}~$\subwalkord$ as the reflexive and transitive closure of~$F$.
\end{definition}

It is straightforward to check that~$\subwalkord$ is a partial order on~$\walks$.  Moreover, one may show that $\subwalkord$ is a suitable order.
We call \defemph{subwalk-minimal} and denote by~$\subwalkminsem$ the semantics based on~$\subwalkord$.

\begin{lemma}\label{l:subwalk-ord-wpo}
    For every database~$\Dc$, $\subwalkord$ is a wpo on~$\walks(\Dc)$
\end{lemma}

Lemma~\ref{l:subwalk-ord-wpo} follows from a corollary of Higman's Lemma \cite{Higman1952} (see also \cite[Sec.~2.2]{Karandikar2015}). Both are stated below.

\begin{lemma}[Higman's Lemma \cite{Higman1952}]
    \label{l:higman}
    If $\preceq$ is a wpo on some set~$X$, then $\overset{\ast}{\preceq}$ on $X^*$ is a wpo, where~$\overset{\ast}{\preceq}$ is the extension of~$\preceq$ to sequences, that is $a_1\cdots a_k \overset{\ast}{\preceq} b_1\cdots b_\ell$
    iff there exists a strictly increasing function $\varphi \colon \{1,\dots,k\}\to\{1,\dots,\ell\}$ such that $a_i\preceq\varphi(a_i)$, for every~$\{1,\dots,k\}$.
\end{lemma}

Note that Higman's Lemma is usually stated for \emph{well-quasi-orders}, which are allowed not to be antisymmetric. The result also holds for well-partial-orders, with an identical proof.
We are interested in a special case of Lemma~\ref{l:higman} stated next,
namely the case where~$X$ is a finite set and~$\preceq$ is the equality relation on~$X$.

\begin{corollary}\label{c:higman}
    Let~$X$ be a finite set, the subsequence relation over~$X^*$ is a wpo.
\end{corollary}

\begin{proof}[Proof of Lemma~\ref{l:subwalk-ord-wpo}] 
    Let~$\Dc$ be a database. We fix~$X=V_\Dc\cup E_\Dc$, a finite set, 
    and let~$\preceq$ denote the subsequence relation over~$X^*$.
Corollary~\ref{c:higman} states that~$\preceq$ is a wpo on~$X^*$.
Note that~$\walks(\Dc)\subseteq X^+$ and it may be checked that~$\subwalkord$ is in fact a restriction of~$\preceq$, that is
    equal to $\setst*{(u,v)\in \walks(\Dc)\times\walks(\Dc)}{u\mathbin{\preceq}v}$.
    This means in particular that the subwalk order~$\subwalkord$ is a wpo on $\walks(\Dc)$.
\end{proof}

The reason for considering the subwalk order~$\subwalkord$ and the associated semantics~$\subwalkminsem$ stems from the property of subwalk coverage guarantee, which we define later. 
Preliminary results about the complexity of query evaluation under~$\subwalkminsem$ can be found in \cite{Khichane2024}, and are recalled in Figure~\ref{f:complexity-table}, page \pageref{f:complexity-table}.
To conclude this section, let us mention another interesting order-based semantics.

\begin{example}[Shortlex Semantics \shortlexsem]\label{ex:shortlex}
    \semindex{\shortlexsem}{Shortlex semantics}
Let~$\preceq_\mathrm{lex}$ be the lexicographic ordering of walks based on some arbitrary order over~$\vertices\cup\edges$.
Let~$\preceq_\mathrm{shortlex}$ be the the corresponding length-lexicographic order: $w \preceq_\mathrm{shortlex} w'$ if either $|w| < |w'|$ or $(|w| = |w'|$ and $w \preceq_\mathrm{lex} w')$.
Even though $\preceq_\mathrm{lex}$ is not suitable, it may be verified that~$\preceq_\mathrm{shortlex}$ is. 
The corresponding shortlex-minimal semantics \shortlexsem{} is a possible implementation of the \gql{ANY} \gql{SHORTEST} keyphrase in GQL~\cite{DeutschEtAl2022,GQL-ISO}.
\end{example}

\subsection{Examples of some out-of-scope semantics}
\label{s:out-of-scope-semantics}

This document focuses on semantics that output walks (Def.~\ref{s:semantics}) and are walk-based (Def.~\ref{d:walk-based}), but it is of course not the only possible approach to defining meaningful RPQ semantics.
The most studied semantics for RPQs in scientific literature is \emph{homomorphism semantics}~\cite{AnglesEtAl2017}, which outputs the endpoints of matches only,
hence does \emph{not} fit our notion of semantics.
On the other hand, \cite{DavidFrancisMarsault2023} propose \emph{run-based} semantics, i.e. semantics~$S$ which require information on the way a walk was matched by the query in order to decide whether it belongs to~$S(\Dc,R)$. We present one such run-based semantics below.
Afterwards, we define another semantics that is not walk-based.
It is loosely inspired by the keyword \gql{CHEAPEST}, discussed during the elaboration of GQL in internal documents, but was not retained for version 1 (see language opportunity 052 in \cite{SQLPGQ-ISO,GQL-ISO}).

\begin{example}[Binding-trail semantics~$\bindingtrailsem$ \cite{DavidFrancisMarsault2023}]
    \semindex{\bindingtrailsem}{Binding-trail semantics}
    \semindex{\simplerunsem}{Simple-run semantics}
    \label{ex:binding-trail} 
    
    Let $\Dc$ be a database,~$R$ a query, and~$k$ is the number of atoms in~$R$.
    A \emph{binding walk (in $\Dc$ matching~$R$)} is a walk~$w\in\matches(\Dc,R)$ together with a function~$\{1,\ldots,\len(w)\}\to\{1,\ldots,k\}$ that maps each edge-position in $w$ to the position of the atom in~$R$ that matched that  edge.
    A  \emph{binding trail (in $\Dc$ matching~$R$)} is a binding-walk $(w,\alpha)$ such that there are no distinct $i,j$ such that~$\alpha(i)=\alpha(j)$ and~$w[i]=w[j]$.
    Binding-trail semantics, denoted by~$\bindingtrailsem$, is defined as follows.
    \begin{equation}
        \bindingtrailsem \colon (\Dc,R) \mapsto \setst{w}{\exists \alpha\text{~~s.t.~~}(w,f)\text{ is a binding trail in $\Dc$ matching~$R$}}
    \end{equation}
    See \cite{DavidFrancisMarsault2023} for a more formal definition.
    A similar semantics called \emph{simple-run semantics}, based on the states of a finite automaton representing the query, is defined in the same paper.
\end{example}

\begin{example}[Cheapest-Walk Semantics $\cheapestwalksem_c$]
    \label{ex:cheapest}
    \semindex{\cheapestwalksem}{Cheapest semantics}
     Let~$c \colon \labels\to(\N-\{0\})$ be a fixed cost function, which we lift to words by:~$c(a_1\cdots a_k) = \sum_{i=1}^{k} c(a_i)$.
     Then, the cost of a walk in a database~$\Dc$ and query~$R$ is the cost of its label in~$\Dc$.
     Cheapest-walk semantics, $\cheapestwalksem_c$, is defined like shortest-walk semantics, using the cost of matches instead of their length.
\end{example}
     
     Note that $\cheapestwalksem_c$ is not walk-based (unless the cost function~$c$ is constant), because a walk provides no information about the label its edges bear.

\section{Properties}
\label{s:properties}

In this section we define and investigate several possible properties of RPQ semantics, in order to help compare and classify them. 
After listing a few basic properties, we distinguish properties relative to some forms of monotonicity, continuity, compatibility with rational operators, and coverage.
Then, we systematically investigate the inclusions between sets of result for the same query according to various semantics.
We conclude this section with a general discussion regarding the set of properties any ``reasonable'' semantics should possess.

\subsection{Basics}
\label{s:basic-properties}

We first list a few basic properties, which are almost always required for a semantics to be considered useful in our model.

\begin{definition}\label{d:identifier-independent}\label{d:label-independent}
    \index{Identifier independent}
    \index{Label independent}
    A semantics~$S$ is called \defemph{identifier-independent} if~$\nu(S(\Dc,R))=S(\nu(\Dc),R)$, for every renaming~$\nu$, database~$\Dc$ and regular expression~$R$.
    Similarly, it is \defemph{label-independent} if $S(\Dc,R)=S(\lambda(\Dc),\lambda(R))$
    for every relabeling~$\lambda$, database~$\Dc$ and regular expression~$R$.
\end{definition}

For instance, a semantics selecting results based on some arbitrary ordering on vertex and edge identifiers (e.g., \shortlexsem{} in Ex.~\ref{ex:shortlex}), or based on the presence of similar arbitrary criteria, would \emph{not} be identifier-independent.
Identifier-independence is related to \emph{genericity} in the classical database setting (see \cite{AbiteboulHullVianu1995} for instance). Indeed, if we consider a query~$q$ consisting of a regular expression and a semantics~$S$ to interpret it, then~$S$ is identifier-independent if and only if~$q$ is $\emptyset$-generic.

Similarly, label-dependent semantics select their results based on arbitrary information regarding labels. Moreover, no label-dependent semantics is walk-based, as stated below.

\begin{lemma}
    \label{l:walk-based=>label-independent}
    Every walk-based semantics is label-independent.
\end{lemma}

\begin{proof}
    By definition, for any walk-based semantics $S$, we have $S(\Dc, R) = r_S(\matches(\Dc,R))$ for some function $r_S$. However, since $\matches(\Dc,R)$ retains no information on the labeling of walks, for any relabeling $\lambda$, $\matches(\lambda(\Dc), \lambda(R)) = \matches(\Dc,R)$. Therefore $S(\Dc, R) = S(\lambda(\Dc), \lambda(R))$.
\end{proof}

In this work, we will always be considering identifier- and label-independent semantics, unless explicitly stated otherwise.

\begin{definition}
    \index{Expression independent}
    A semantics~$S$ is \defemph{expression-independent} if, for every database~$\Dc$ and every expressions~$R_1,R_2$ such that~$L(R_1)=L(R_2)$, we have $S(\Dc,R_1)=S(\Dc,R_2)$.
\end{definition}

Expression-independent semantics select their results on the sole basis of the language denoted by the regular expression. In some settings, it may make sense to instead let the selection of results depend on the structure of the expression itself. Binding-trail semantics (Ex.~\ref{ex:binding-trail}) is one such example. However, this is out of the scope of this work, since all walk-based semantics are expression-independent, as stated below.

\begin{lemma}
    \label{l:walk-based=>expression-independent}
    Every walk-based semantics is expression-independent.
\end{lemma}

\begin{proof}
    Similar to Lemma \ref{l:walk-based=>label-independent}. If $L(R_1)=L(R_2)$ then $\matches(\Dc,R_1) = \matches(\Dc,R_2)$ by definition of $\matches$. 
    Therefore $S(\Dc, R_1) = S(\Dc, R_2)$.
\end{proof}

\begin{definition}\label{d:bounded-semantics}\label{d:unbounded-semantics}
    \index{Bounded}\index{Unbounded}
    A semantics~$S$ is called \defemph{unbounded} if, for every integer~$n$, there is a database~$D$ and a regular expression~$R$ such that~$S(D,R)$ contains a walk of length~$n$ or more.
    Otherwise~$S$ is called \defemph{bounded}.
\end{definition}

Of course, we expect any reasonable semantics to be unbounded since a bounded semantics never returns a walk longer than some constant.

\subsection{Monotony}
\label{s:monotony}

\begin{definition}[Monotony, co-monotony and restrictibility]
        \index{Monotonous}
        A semantics~$S$ is said to be \defemph{monotonous} if $S(\Dc,R)\subseteq S(\Dc',R)$, for every regular expression~$R$ and databases~$\Dc,\Dc'$ such that~$\Dc\subseteq\Dc'$.
\label{d:comonotonous}
        \index{Co-monotonous}
It is \defemph{co-monotonous} if $S(\Dc,R)\supseteq \big(S(\Dc',R)\cap\walks(D)\big)$, for every regular expression~$R$ and databases~$\Dc,\Dc'$ such that~$\Dc\subseteq\Dc'$.
It is \defemph{restrictible} if both hold.
        \label{d:restrictible}        
        \index{Restrictible}
\end{definition}

These notions are about the behavior of the semantics when a database is enriched with new vertices and edges.
A monotonous semantics does not remove results when the database is enriched. A co-monotonous semantics does not add results that are unrelated to the new elements.
Finally, a restrictible semantics means that one is able to compute part of the result, given part of the database.
A simple verification shows that all semantics defined above are co-monotonous. This is in fact true of all filter-based and order-based semantics. Additionally, filter-based semantics are also monotonous, and therefore restrictible.

\begin{lemma}
    \label{l:filter-based=>restrictible}
    \label{l:order-based=>comonotonous}
    Every filter-based semantics is restrictible, and
    every order-based semantics is co-monotonous.
\end{lemma}

\begin{proof}[Proof for filter-based semantics]
    Let $S$ be a filter-based semantics defined by filter function $f_S$ (see Def. \ref{d:filter-based}). 
    Let $R$ be a regular expression and~$\Dc,\Dc'$ two databases such that~$\Dc\subseteq\Dc'$. 
    We have:
    \begin{align*}
        S(\Dc', R) ={}& \setst{w\in\matches(\Dc',R)}{f_S(w)=\top}
\end{align*}
    Therefore:
    \begin{equation*}
        S(\Dc', R) \cap \walks(\Dc) = \setst{w\in\matches(\Dc,R)}{f_S(w)=\top} = S(\Dc, R)\qedhere
    \end{equation*}
\end{proof}

\begin{proof}[Proof for order-based semantics]
    Let $S$ be an order-based semantics defined by ordering $\preceq$ (see Def. \ref{d:order-based}). 
    Let $R$ be a regular expression and~$\Dc,\Dc'$ two databases such that~$\Dc\subseteq\Dc'$. 
    Let $w$ be a walk between some vertices $s$ and $t$ in $S(\Dc',R)\cap\walks(D)$.
    By definition, $w \in \min_\preceq(\matches(\Dc',R,s,t))$. But since $w \in \walks(\Dc)$, we also have $w \in \matches(\Dc, R)$. There can be no $w' \prec w$ in $\matches(\Dc, R)$, otherwise $w$ would not be minimal in $\matches(\Dc',R,s,t)$ either. Therefore $w \in \min_\preceq(\matches(\Dc,R,s,t))$, in other words $w \in S(\Dc, R)$.
\end{proof}

We expect all ``reasonable'' semantics to be co-monotonous.
In fact, semantics that are not co-monotonous are usually quite strange, as illustrated below.

\begin{example}[Log-length semantics]
    \semindex{\loglensem}{Log-length semantics}
    The \emph{log-length} semantics~$\loglensem$ returns all matches that are
    shorter than~$\log_2(|\Dc|)$. 
    This semantics has the property that the size of~$\loglensem(\Dc,R)$ is always polynomial in the size of~$\Dc$.
    It is not co-monotonous since doubling the size of~$\Dc$ (for instance by adding enough isolated vertices) may return longer walks which do not necessarily use the new isolated vertices.
    Note however that $\loglensem$ is monotonous, since no previous result may disappear when enriching the database.

    Although $\loglensem$ is not walk-based, one could consider a variant that is.  Indeed,
    it could return the walks shorter than~$\log_2(|X(\Dc,R)|)$, where~$X(\Dc,R)$ is the set of all edges appearing in some walk in~$\matches(\Dc,R)$.
\end{example}

As for monotony, the landscape is more interesting.
We already saw that all filter-based semantics are monotonous (Lemma \ref{l:filter-based=>restrictible}).
However, it is notable that shortest semantics is not.

\begin{lemma}\label{l:shortest-not-restrictible}
    $\shortestwalksem$ is not monotonous.
\end{lemma}
\begin{proof}
    Let~$\Dc$ be a database consisting of three vertices~$v_1$, $v_2$ and $v_3$ and two edges~$v_1\rarrow[a]{e_1}v_2$ and~$v_2\rarrow[a]{e_2}v_3$,
    and let~$R=aa+b$ where~$a,b\in\labels$.
    Hence~$S(\Dc,R)=\set{v_1\rarrow{e_1}v_2\rarrow{e_2}v_3}$.
    Then, let~$\Dc'$ be the database~$\Dc'$ resulting from adding in~$\Dc$ a new fresh edge~$v_1\rarrow[b]{e_2}v_3$.
    Hence~$S(\Dc',R)=\set{v_1\rarrow[b]{e_2}v_3}$ which is not a superset of~$S(\Dc,R)$.
\end{proof}

At first glance, it may seem that all order-based semantics~$S_\preceq$ are bound to be non-monotonous, since it is enough to make a shortcut (w.r.t~$\preceq$) using new vertices and edges.
However, one may show that the (order-based) subwalk-minimal semantics \subwalkminsem{} (Def.\@ \ref{d:subwalk-order}) is indeed monotonous.
More generally, let us now characterize monotonous order-based semantics using a notion of element inclusion (Proposition~\ref{p:order-based+monotonous<=>inclusion}).
In the following, we denote by~$\vertexset(w)$ (resp.~$\edgeset(w)$) the set of vertices (resp.~of edges) occurring in~$w$.  We also use~$\elemset(w)$ (resp.~$\elembag(w)$) to denote the set (resp.~the bag) of vertices and edges appearing in~$w$.
We say that an order~$\preceq$ on~$\walks$ \defemph{respects element inclusion} if $w\preceq w'$ implies~$\elemset(w)\subseteq \elemset(w')$.

\begin{proposition}\label{p:order-based+monotonous<=>inclusion}
    Let~$S_\preceq$ be a semantics based on a trimmed order~$\preceq$.
    Then~$\preceq$ respects element inclusion iff $S_\preceq$ is monotonous.
\end{proposition}

\begin{proof}[Proof of Proposition~\ref{p:order-based+monotonous<=>inclusion}, forward direction]
    For the sake of contradiction, let us assume there are a query~$R$, two databases~$\Dc\subseteq\Dc'$ and a walk~$w$ that belongs to~$S(\Dc,R)$ but not to~$S(\Dc',R)$.
    This means there is a walk~$w'\in S(\Dc',R)$ satisfying $w'\prec w$.
    It follows from hypotheses that $\elemset(w')\subseteq \elemset(w)$
    hence that~$w'\in\walks(\Dc)$.
    Since~$S$ is co-monotonous (Lemma~\ref{l:order-based=>comonotonous})
    $w'\in S(\Dc,R)$. Hence $w\notin S(\Dc,R)$, a contradiction.
\end{proof}

The proof of the backward direction requires some extra machinery (Definition~\ref{d:characteristic}) that will also be useful in future proofs.

\begin{definition}[Characteristic database, characteristic expression] 
    \label{d:characteristic}
    For simplicity, we will consider in the following that~$\edges\subseteq\labels$, i.e.\@ that edge identifiers are possible labels.
    To be formally correct, we would need to use some arbitrary bijection~$\edges\to\labels$, but we prefer to leave it implicit.
    \begin{subthm}
        \item 
             \label{d:characteristic-database}
             \index{Characteristic database}
        Let~$w_1,\ldots,w_k$ be consistent walks.
        Let~$\Dc$ be any database such that~$\{w_1,\ldots,w_k\}\subseteq\walks(\Dc)$, and
        \begin{align*} 
            V={}&\vertexset(w_1)\cup\cdots\cup \vertexset(w_k) \\
            E={}&\edgeset(w_1)\cup\cdots\cup \edgeset(w_k)
        \end{align*}
        Note that $V\subseteq V_\Dc$ and $E\subseteq E_\Dc$. 
        The \defemph{characteristic database} of~$\{w_1,\ldots,w_k\}$ is the database~
        \begin{equation*}
            \Cc_{w_1,\ldots,w_k} = (L, V, E, \src, \tgt, \lbl)
        \end{equation*} 
        where $L = E$,  $\lbl$ is the identity function, 
        and $\src$ (resp.~$\tgt$) is the restriction of~$\src_\Dc$ (resp.~$\tgt_\Dc$) to~$E$.
        
        (One may verify that~$\Cc_{w_1,\ldots,w_k}$ is independent of the database~$\Dc$ chosen.)
        \item  
         \label{d:characteristic-expression}
         \index{Characteristic expression}
        Let~$w$ be a walk of positive length $k$. Then we denote by~$\CharR{w}$ the \defemph{characteristic regular expression} of $w$, defined as~$\CharR{w}=w[1]\cdots w[k]$.
    \end{subthm}
\end{definition}

Intuitively, the characteristic database of a finite set of compatible walks is just the smallest database containing them all, i.e., $\{w_1,\ldots,w_k\} \subseteq \walks(\CharD{w_1,\ldots,w_k})$, and in which edges are labeled with their own identifier. The characteristic expression of a walk is just the sequence of its edge identifiers. 

These definitions are chosen in such a way that the characteristic expression of each $w_i$ of positive length is matched exactly once ($\matches(\CharD{w_1,\ldots,w_k},\CharR{w_i})=\{w_i\}$) and that the characteristic expression of any other walk of positive length is matched either once or not at all: $\matches(\CharD{w_1,\ldots,w_k},\CharR{w})$ is either empty or equal to~$\{w\}$.

\medskip

We can now show that if an order-based semantics is monotonous, then the corresponding order preserves element inclusion.

\begin{proof}[Proof of Proposition~\ref{p:order-based+monotonous<=>inclusion}, backward direction]
    Let~$S_\preceq$ be a monotonous semantics based on a trimmed order~$\preceq$.

    Let~$w$ and~$w'$ be two walks such that~$w\mathrel{\prec} w'$.
    First, let us address the case where~$\len(w)=0$, in which case~$\elemset(w)=\{v\}$ for some vertex~$v$.
    Since~$\preceq$ is trimmed,~$w$ and~$w'$ have the same endpoints. It follows that~$v=\src(w')=\tgt(w')$ hence~$\elemset(w)\subseteq\elemset(w')$.
    
    Second, assume that~$\len(w)>0$.
    Since~$\preceq$ is trimmed,~$w$ and~$w'$ are consistent,
    which means in particular that the two databases~$\CharD{w'}$ and $\CharD{w,w'}$ are well defined.
    Let us moreover consider the following query
    \begin{equation}
        R=\begin{cases}
            \CharR{w}+\CharR{w'} & \text{if $\len(w')\mathbin>0$} \\
             \CharR{w}+\varepsilon  &\text{if $\len(w')\mathbin=0$} \\
        \end{cases}
    \end{equation}
    Note that~$\matches(\CharD{w,w'},R)\supseteq\{w,w'\}$, and that the inclusion might be strict when $\len(w') = 0$ because $\varepsilon$ will match every walk of length 0. 
    Since $w \prec w'$, this implies that $w'\notin S(\CharD{w,w'},R)$,
    and therefore, due to the monotony of~$S$, that 
    \begin{equation}
        w'\notin S_{\preceq}(\CharD{w'},R)\,.
    \end{equation}
    
    If $w'$ were of length~$0$, then $\CharD{w'}$ would consist of a single isolated vertex.
    Therefore, since $\len(w)>0$, $w$ is not a walk in $\CharD{w'}$. Hence we would have~$\matches(\CharD{w'},R)=\{w'\}$ and thus $w'\in S_{\preceq}(\CharD{w'},R)$, a contradiction to (\theequation).
    It follows that~$w'$ must have a positive length, which means that~$\matches(\CharD{w,w'},R)=\{w,w'\}$. 
    Since $\matches(\CharD{w'},R)\subseteq\matches(\CharD{w,w'},R)=\{w,w'\}$,
    Equation (\theequation) implies that $w$ must belong to~$\matches(\CharD{w'},R)$ hence to $\walks(\CharD{w'})$.
    By definition of the characteristic database, this is true only if $\elemset(w)\subseteq\elemset(w')$.
\end{proof}

Consider the relation~$\preceq_1$ defined as: $w\preceq_1 w'$ if $\elemset(w)\subseteq \elemset(w')$.
It is not antisymmetric, hence not a partial order.
Now consider the relation~$\preceq_2$ defined as: $w\prec_2 w'$ if $\elemset(w)\subsetneq \elemset(w')$.
Although~$\preceq_2$ is a partial order, it is \textbf{not} suitable:~$\preceq_2$ is not a wpo on~$\walks(D)$ whenever~$\Dc$ contains a cycle.
We give below two examples of semantics based on an order that preserves element inclusion.
Example~\ref{ex:minimal-multiset} takes multiplicity into account.
Example~\ref{ex:shortest-minimal} refines~$\preceq_2$ by adding a tiebreaker for walks with identical element sets.

\begin{example}[Minimal-multiset semantics~\minmultsem]\label{ex:minimal-multiset}
    \semindex{\minmultsem}{Minimal-multiset semantics}
    Let~$\minmultsem$ be the semantics based on the order~$\preceq$ defined by:~$w\prec w'$ if $\elembag(w)\subsetneq \elembag(w')$.
    Recall that~$\subwalkord$ denotes the subwalk order on walks (Def.~\ref{d:subwalk-order}). 
    A simple verification yields that~$w\subwalkord w' \implies w\preceq w'$.
    Then, $\preceq$ is a wpo on $\walks(\Dc)$ for each database~$\Dc$, because so is~$\subwalkord$ (Lemma~\ref{l:subwalk-ord-wpo}).
    Indeed, let~$(w_1,w_2,\ldots)$ be an infinite sequence of walks in~$\Dc$.
    Since~$\subwalkord$ is a wpo on~$\Dc$, there are $i<j$ such that~$w_i \subwalkord w_j$, which
    implies that~$w_i \preceq w_j$.
\end{example}

\begin{example}[Shortest-minimal-set semantics~\shortminsetsem]\label{ex:shortest-minimal}
    \semindex{\shortminsetsem}{Shortest-minimal-set semantics}
    Let~$\shortminsetsem$ be the semantics based on the order~$\preceq$ defined by:
    \begin{equation*}
        w\prec w' \iff \begin{cases}
            \text{either }\elemset(w)\subsetneq\elemset(w') \\
            \text{or }\elemset(w)=\elemset(w')\text{ and }\len(w)<\len(w')
        \end{cases}
    \end{equation*}
\end{example}

\noindent
The semantics~$\minmultsem$ and~$\shortminsetsem$ are finite, identifier-independent and restrictible.
The complexity of query evaluation under them is explored in \cite{Marsault2024} and recalled in Figure~\ref{f:complexity-table}, page~\pageref{f:complexity-table}.

\subsection{Continuity}
\label{s:continuity}

In this section, we attempt to capture the property for a semantics $S$ to behave ``reasonably'' at the limit.
Essentially, the sequence~${\big(S(\Dc,a+a^2+\cdots+a^k)\big)_{k\in\N}}$ should to tend to~$S(\Dc,a^*)$ when~$k$ tends to~$\infty$.
More precisely, results of non-decreasing sequences of RPQs (with respect to language inclusion) should eventually stabilize to some fixed set $W$, and furthermore whenever the whole sequence is captured by a single RPQ, then the set of results over that RPQ should be $W$ itself. 
These conditions prevent $S$ to change gears and arbitrarily alter the set of results when moving to the limit.
As we show below, it turns out that all filter-based and all order-based semantics possess these properties.

\begin{definition}
    \label{d:stabilizing}
    \label{d:continuous}
    \index{Stabilizing}
    \index{Continuous}
    A semantics~$S$ is said to be \defemph{stabilizing} if, for every database~$\Dc$ and every sequence of regular expressions~$\big(R_i\big)_{i\in\N}$ such that $\big(L(R_i)\big)_{i\in\N}$ is non-decreasing, $\big(S(\Dc, R_i)\big)_{i \in \N}$ is ultimately equal to some set $W$, and we say that $\big(S(\Dc, R_i)\big)_{i \in \N}$ stabilizes to $W$.

    Moreover, when~$S$ is stabilizing, it is said to be \defemph{continuous}
    if, whenever there exists a sequence of regular expressions~$\big(R_i\big)_{i\in\N}$ as above and a regular expression~$R_\infty$ such that~$L(R_\infty)=\bigcup_{i\in\N}L(R_i)$, it holds that~$W=S(\Dc, R_\infty)$.
\end{definition}

\begin{proposition}\label{p:order-based=>continuous}\label{p:filter-based=>continuous}
    All order-based and filter-based semantics are continuous.
\end{proposition}

\begin{proof}[Proof for filter-based semantics]
    Let~$S$ be a semantics based on a filter~$f$, and let~$\Dc$ be a database.
    The function~$F$ that maps every~$W\subseteq \walks(\Dc)$ to $\{w\in W\mid f(w)=\top\}$ is monotonous.
    Also note that~$F(\walks(\Dc))$ is finite, otherwise~$S(\Dc,\Sigma^*)$ would not be.
    It follows that~$S$ is stabilizing and continuous.
\end{proof}

The proof for order-based semantics requires more care.
We first prove that all order-based semantics are stabilizing (Lemma~\ref{l:order-based=>stabilizing}).

\begin{lemma}\label{l:order-based=>stabilizing}
    All order-based semantics are stabilizing.
\end{lemma}
\begin{proof}
  Consider an order-based semantics~$S_\preceq$, a sequence $\big(R_i\big)_{i\in\N}$ as in the statement of Definition \ref{d:continuous} and a database $\Dc$, and let $M_i = \matches(\Dc, R_i)$ and $S_i = S_\preceq(\Dc, R_i)$.
    Since~$(M_i)_{i\in\N}$ is non-decreasing and by definition of the $S_i$, we have
    \begin{equation}\label{eq:MiSi-incl}
    \begin{array}{ccccccccc}
    M_0 & \subseteq & M_1 & \subseteq & \cdots & \subseteq &  M_i & \subseteq & \cdots 
    \\
    \rotatebox[origin=c]{90}{$\subseteq$}&&
    \rotatebox[origin=c]{90}{$\subseteq$}&&&&
    \rotatebox[origin=c]{90}{$\subseteq$}&
    \\
    S_0 && S_1 &&&& S_i && 
    \end{array}
    \end{equation}
    This implies in particular that 
    $\displaystyle\bigcup_{j=0}^i S_j \subseteq \bigcup_{j=0}^i M_j \subseteq M_i$, for every $i\in\N$·
 
    First, we claim that for every~$i$ such that~$S_{i+1} \neq S_{i}$, there exists a walk~$w\in (S_{i+1}\setminus S_i)$.
    Since~$S_{i+1} \neq S_{i}$, we have $(S_{i+1}\setminus S_{i}) \neq \emptyset$ or $(S_{i}\setminus S_{i+1}) \neq \emptyset$.
    In the first case, the claim immediately holds.
    In the second case, let~$w'$ be some walk in~$S_{i}\setminus S_{i+1}$.
    Since~$w' \in S_i \subseteq M_i \subseteq M_{i+1}$ but~$w'\notin S_{i+1}$, there must be a walk~$w$ in~$S_{i+1}$ such that~$w\prec w'$.
    This walk cannot belong to~$S_i$, otherwise $w'$ would not either.
    This concludes the proof of the claim.

    Second, we claim that for every integer~$i$ and walk~$w$ such that~$w\in (S_{i+1} \setminus S_i)$, there is no~$w'\in \bigcup_{j=0}^{i} S_j$ such that~$w' \preceq w$.
    By definition, since~$w\in S_{i+1}$, then~$w$ is minimal in~$M_{i+1}$.
    It follows from~\eqref{eq:MiSi-incl} that~$w$ is incomparable to or smaller than every element in~$M_j$ (resp.~in~$S_j$), for every $j\leq i$ ($\ast$).
    It follows that the walk~$w$ cannot be in~$M_i$, otherwise it would also belong to~$S_i$, a contradiction to its definition.
    We then deduce from~\eqref{eq:MiSi-incl} that for all $j \leq i$, $w \notin S_j$ ($\ast\ast$).
    This concludes the proof of the claim: any~$w'$ as in the claim
    would either be equal to~$w$, a contradiction to ($\ast\ast$), or distinct from~$w$, a contradiction to ($\ast$).

    Finally, if $(S_i)_{i\in\N}$ were not eventually constant, we could apply the two claims above on indices~$i$ such that~$S_{i+1}\neq S_{i}$.
    This yields a sequence of walks~$(w_k)_{k\in\N}$ such that for each~$k<k'$,
    it does \textbf{not} hold that $w_k\preceq w_{k'}$,
    a contradiction to the fact that $\preceq$ is a well-partial-order on $\walks(\Dc)$.
\end{proof}

\begin{proof}[Proof of Proposition~\ref{p:order-based=>continuous} 
for order-based semantics]
Let~$S_\preceq$ be an order-based semantics, $\Dc$ a database, $\big(R_i\big)_{i\in\N}$ a sequence of regular expressions such that $\big(L(R_i)\big)_{i\in\N}$ is non-decreasing,
and~$R_\infty$ and expression such that~$L(R_\infty)=\bigcup_{i\in\N}L(R_i)$.
Since~$S_{\preceq}$ is stabilizing, there is a set~$W$ and a bound~$b$ such that for every~$n>b$,~$S(\Dc,R_n)=W$.
We have to show that ~$S(\Dc,R_{\infty}) = W$, which follows from the fact that the following two statements are equivalent for every walk~$w$:
\begin{itemize}
    \item $w\in\matches(\Dc,R_\infty)$ and there exists~$w'\in\matches(\Dc, R_\infty)$ such that $w' \prec w$;
    \item there exists an index~$n>b$ and a walk $w'$ such that $w\in\matches(\Dc,R_n)$, $w'\in\matches(\Dc, R_n)$ and~$w'\prec w$. 
\end{itemize}
This equivalence follows from the definitions of the $R_i$'s and $R_\infty$.
\end{proof}

We expect any ``reasonable'' (walk-based) semantics to be continuous.
Indeed, a non-continuous walk-based semantics must treat finite sets and infinite sets differently.
To give the reader a flavor of what non-continuous semantics might look like, we give below two toy examples: a non-stabilizing semantics and a stabilizing but non-continuous semantics.

\begin{example}[Giving-Up Semantics]\label{ex:giving-up-semantics}
    \semindex{\givingupsem}{Giving-up semantics}
    Most semantics presented in this document are designed to return at least one walk between each pair of vertices for which there is a match, even in cases where there are actually infinitely many such matches.
In contrast, let $\givingupsem(\Dc,R)$ be the semantics which returns the whole set $\matches(\Dc, R)$ if this set is finite, and the empty set of walks otherwise.
    This implies that $\givingupsem(\Dc,R)$ is always a finite set.
    
    Given expressions $\big(R_i = a^{\leq i}\big)_{i\in\N}$ (each matching all walks labeled by some word $a^j$ with $j \leq i$), if $\Dc$ contains an $a$-labeled loop $e$ from some vertex $v$ to itself, then each $\givingupsem(\Dc,R_i)$ contains the new walk $(ve)^iv$. Hence \givingupsem{} is clearly not stabilizing (and thus not continuous: note that in the above case $\givingupsem(\Dc, a^*) = \emptyset$).
        
\givingupsem{} seems too bizarre to be interesting. But one may define variants which still have interesting properties.
    For instance, consider the semantics  $\givingupsem'$, defined as follows.
    For a given set of elements $X$ and two vertices $s$ and $t$, let us denote by $\matches_{X}(\Dc, R, s, t)$ the set of matches $w$ for $R$ in $\Dc$ with $\ep(w) = (s, t)$ and $\elemset(w) \subseteq X$.
    Let~$\mathbb{T}$ be the set of all triples~$(X,s,t)$ such that $\matches_{X}(\Dc, R, s, t)$ is finite. In other words, the walks whose set of elements support infinitely many other matches for $R$ are not returned as results.
    Then, 
    \begin{equation}\label{eq:gu'}
        \givingupsem'(\Dc, R) = \bigcup_{(X,s,t)\in \mathbb{T}} \matches_{X}(\Dc, R, s, t)\,.
    \end{equation}
    Since $\mathbf{T}$ is finite, this semantics is finite by definition. It can be shown that it is in fact also walk-based and restrictible.
\end{example}

Semantics that are stabilizing but non-continuous are even stranger. 
The typical example is a semantics that changes behaviour in the presence of an infinite set.

\begin{example}[Stabilizing non-continuous semantics]
\label{ex:weird-semantics}
Let $S$ be the semantics defined by:
\begin{equation}
    S(\Dc,R) = 
    \begin{cases}
        \shortestwalksem(\Dc,R) &\text{if }\matches(\Dc, R)\text{ is finite},\\
        \trailsem(\Dc,R) &\text{if }\matches(\Dc, R)\text{ is infinite}.
    \end{cases}
\end{equation}
It is stabilizing because, for every sequence of regular expressions~$\big(R_i\big)_{i\in\N}$ such that $\big(L(R_i)\big)_{i\in\N}$ is non-decreasing,
the sequence~$(\matches(\Dc, R_i))_{i\in\N}$ either contains
only finite sets, or there is a bound~$b$ such that $\matches(\Dc, R_i)$ is finite only if~$i\leq b$.
In both cases, $(S(\Dc, R_i))_{i\in\N}$ stabilizes because both \shortestwalksem{} and \trailsem{} do.
It is not continuous, as can be seen by considering the sequence~$(a^{\leq i})_{i\in\N}$ and the limit expression~$a^*$.

\end{example}

Other examples of stabilizing non-continuous semantics can be built as variants of other semantics in which we add a condition about the infiniteness of the set of matches.
For instance, given any order-based semantics~$S_\preceq$,
the semantics~$S'_\preceq$ defined below returns only those walks which are smaller than infinitely many walks with respect to $\preceq$:
\begin{multline*}
S'_\preceq(\Dc,R) =
\{ w \in S_{\preceq}(\Dc,R) \mid \exists\text{ infinitely many }w'\mathbin\in\matches(\Dc,R)\text{ such that }w\mathbin\preceq w' \}
\end{multline*}
It is stabilizing because $S_{\preceq}$ is, and the usual sequence of expressions $(a^{\leq i})_{i\in\N}$ with limit expression $a^*$ shows that it is not continuous.

\subsection{Compatibility with rational operators}
\label{s:rational operator compatibility}

In this section we study interactions between the structure of an RPQ as a regular expression and the set of results. 
We define in particular the notions of composability and decomposability of a semantics with respect to a given RPQ operator.
The intuition behind composability, stated in the case of concatenation, is as follows.
A semantics is $\cdot$-decomposable if any result for~$R\cdot R'$, that is a walk~$w$ in $S(D,R \cdot R')$, can be \emph{decomposed} into a result for~$R$ and a result for~$R'$. This means that there exist~$w' \in S(D,R)$ and~$w'' \in S(D,R')$ such that~$w = w' \cdot w''$.
Conversely, a semantics $S$ is $\cdot$-composable if the results  for~$R$ and~$R'$ can be \emph{composed} freely into (some of the) results for $R\cdot R'$.
These notions are defined in more detail below.

\begin{definition}
    \index{Composable!w.r.t.~rational operators}
    \label{d:composability}
    \index{Decomposable!w.r.t.~rational operators}
    \label{d:decomposability}
    \index{Compatible!with rational operators}
    We say that~$S$ is~\defemph{$\cdot$-composable} (resp.~\defemph{$\cdot$-decomposable}) if Equation~\eqref{eq:.-composable} (resp.~Equation~\eqref{eq:.-decomposable}), below, holds for every database $\Dc$ and RPQs~$R,R'$.
    \begin{align}
    \label{eq:.-composable}
    S(\Dc,R \cdot R') \supseteq{}& S(\Dc,R)\cdot S(\Dc,R')
     \\
    \label{eq:.-decomposable}
      S(\Dc,R \cdot R') \subseteq{}& S(\Dc,R)\cdot S(\Dc,R')
    \end{align}
If both~\eqref{eq:.-composable} and~\eqref{eq:.-decomposable} hold, we say that~$S$ is \defemph{compatible with~$\cdot$}, or \defemph{$\cdot$-compatible} for short.
Similarly, we define \defemph{$+$-composability}, \defemph{$+$-decomposability}, \defemph{$+$-compatibility} from \eqref{eq:+-composable}, below.
\begin{equation}
    \label{eq:+-composable}
    S(\Dc,R + R') \supseteq{} S(\Dc,R)\cup S(\Dc,R') 
    \quad\text{and/or}\quad
S(\Dc,R + R') \subseteq{} S(\Dc,R)\cup S(\Dc,R') 
\end{equation}
and \defemph{$\ast$-composability}, \defemph{$\ast$-decomposability}, \defemph{$\ast$-compatibility} from \eqref{eq:*-composable} below.
\begin{equation}
    \label{eq:*-composable}
    S(\Dc,R^*) \supseteq \left(S(\Dc,R)\right)^* 
 \quad\text{and/or}\quad
S(\Dc,R^*) \subseteq \left(S(\Dc,R)\right)^* 
\end{equation}
Finally, we say that a semantics is \defemph{decomposable} if it is decomposable w.r.t.~$\cdot$, $+$ and~$*$.
\end{definition}

\begin{definition}\label{d:atom-compatibility}\index{Compatible!with atoms}
We say that a semantics~$S$ is \defemph{compatible with atoms}, or $a$-compatible, if~$S(\Dc,a) = \matches(\Dc,a)$, for every $a\in\labels$ and database~$\Dc$.
\end{definition}

We give below several statements regarding compatibility with rational operators.
First, Lemma \ref{l:compatibilites-plus} states that all filter-based and order-based semantics are $+$-decomposable, and that $+$-composability is essentially equivalent to being filter-based.

\begin{lemma}[Compatibility with $+$]\label{l:compatibilites-plus}
    \begin{subthm}
        \item \label{l:filter-based=>+compatible} 
        Every filter-based semantics is $+$-compatible.
        
        \item \label{l:order-based=>+incompatible} 
        Every order-based semantics is $+$-decomposable but not~$+$-composable.

        \item \label{l:+composable=>filter-based}
        For every $+$-compatible semantics~$S$, there exists a filter-based semantics~$S'$ such that for every database~$\Dc$, regular expression~$R$
        and walk~$w$ \textbf{of positive length}: $w\in S(\Dc,R)$ iff $w\in S'(\Dc,R)$.
    \end{subthm}
\end{lemma}

\begin{proof}[Proof of Item~\ref{l:filter-based=>+compatible}]
    This directly follows from definitions.
    Let~$S_\preceq$ be a filter-based semantics defined by filter $f$. 
    For $+$-decomposability, consider any $R$, $R'$ and $\Dc$.
    If some walk $w$ labeled by $u$ is in $S(\Dc, R + R')$, then $f(w)$ holds, and either $u \in L(R)$ or $u \in L(R')$. Therefore either $w \in S(\Dc, R)$ or $w \in S(\Dc, R')$ (or both).\\
    For $+$-composability, consider RPQs $R$ and $R'$ and some database $\Dc$. For any walk $w \in S(\Dc, R)$, i.e.\@ such that $\lbl(w) \in R$ and $f(w)$ holds, $w$ is also, by definition, in $S(\Dc, R+R')$ since $\lbl(w) \in L(R+R')$. This holds similarly for any $w' \in S(\Dc, R')$.
\end{proof}
    
\begin{proof}[Proof of Item~\ref{l:order-based=>+incompatible}]
    Let~$S_\preceq$ be an order-based semantics. For $+$-decomposability, consider any $R$, $R'$ and $\Dc$.
    If some walk $w$ labeled by $u$ is in $S(\Dc, R + R')$, then $w$ is minimal w.r.t. $\preceq$ among walks in $\matches(\Dc, R + R')$. Since both $\matches(\Dc, R)$ and $\matches(\Dc, R')$ are subsets of $\matches(\Dc, R + R')$, $w$ is not strictly greater than any walk in those sets. Furthermore, $u \in L(R)$ or $u \in L(R')$. Therefore, $w \in S(\Dc, R)$ or $w \in S(\Dc, R')$. \\
    To show that $S$ is not $+$-composable, consider a database~$\Dc$ consisting of a vertex~$v$ and a loop~$v\rarrow{e} v$ labeled by some~$a\in\labels$.
    We let~$w$ denote the walk~$v\rarrow{e} v$, hence~$w^i$ is well defined.
    Since~$\preceq$ is a wpo on~$\walks(\Dc)$ then there exist distinct~$i,j$ such that~$w^i \prec w^j$.  Hence~$S(\Dc,a^i+a^j)= \{a^i\} \not \supseteq eq \big(S(\Dc,a^i)\cup S(\Dc,a^j)\big)$.
\end{proof}

\begin{proof}[Proof of Item~\ref{l:+composable=>filter-based}]
    Let us define the filter~$f$ as follows.
    We use the notion and notation of characteristic database~$\CharD{w}$ and regular expression~$\CharR{w}$ (Def.~\ref{d:characteristic}).
    Each walk~$w$ that is \textbf{not} self-consistent is mapped to~$\bot$.
    Any self-consistent walk~$w$ is mapped to~$\top$ if~$S(\CharD{w},\CharR{w})=\{w\}$,
    and to~$\bot$ otherwise.
    Let~$D$ be a database and~$R$ an RPQ. Let us show that a walk~$w\in\matches(\Dc,R)$ of positive length is in~$S(\Dc,R)$ if and only if $f(w)=\top$.

    Let~$w$ be a walk of positive length in $\matches(\Dc,R)$ and we write~$u=\lbl_{\Dc}(w)$. The language~$L(R)-\{u\}$ is regular hence is denoted by a regular expression~$R'$.
    Since~$S$ is walk-based, we have the following.
    \begin{equation}
        w\in S(\Dc,R) \iff w\in S(\Dc,R'+u)
    \end{equation}
    Note in particular that~$w\notin \matches(\Dc,R')$,
    hence~$+$-compatibility of~$S$ yields the following.
    \begin{equation}
        w\in S(\Dc,R'+u) \iff w\in S(\Dc,u)
    \end{equation}
    Let~$w_1,\ldots,w_k$ be all the (finitely many) walks in~$\matches(\Dc,u)$; hence $w_i=w$ for some~$i$.
    All these walks have the same length, and since~$w$ is among them, they have positive length.
    Let us now consider the database~$\Dc''=\CharD{w_1,\ldots,w_k}$, and regular expression~$R''=\sum_{i=0}^{k} \CharR{w_i}$.
    In particular,~$\matches(\Dc'',R'')=\{w_1,\ldots,w_k\}$
    hence~$S(\Dc,u)=S(\Dc'',R'')$, which implies the following.
    \begin{equation}
        w\in S(\Dc,u) \iff w\in S(\Dc'',R'')
    \end{equation}
    Finally, an easy induction using~$+$-compatibility of~$S$ yields the following.
    \begin{equation}\label{l:+composable=>filter-based-III}
        w\in S\left(\CharD{w_1,\ldots,w_k}, \sum_{i=0}^{i} \CharR{w_i}\right) \iff w\in S\left(\CharD{w_1,\ldots,w_k}, \CharR{w}\right) \iff f(w)=\top
    \end{equation}
\end{proof}

\begin{remark}
    The equivalence in the statement of Item~\ref{l:+composable=>filter-based} is not true for walks of length~$0$.
    Indeed, consider a semantics~$S$ defined as follows. We denote by~$\walks_0(\Dc)$
    the walks in~$\walks_0(\Dc)$ of length~$0$.  Then, we define $S(\Dc,R)=\walks_0(\Dc)$
    if~$\varepsilon\in L(R)$ and~$\Dc$ has an even number of vertices; or~$S(\Dc,R)=\emptyset$ otherwise.
    It may be checked that~$S$ is walk-based and $+$-compatible, although it is clearly not filter-based.
\end{remark}

As hinted previously, the two families of filter-based and order-based semantics are disjoint. It
indeed follows from items~1 and~2 of Lemma \ref{l:compatibilites-plus}.

\begin{corollary}
    There is no semantics~$S$ that is both order-based and filter-based.
\end{corollary}

Next, we state in Lemmas \ref{l:concat-composability} and \ref{l:star-composability} that no order-based or filter-based semantics is $\cdot$-composable or $\ast$-composable, under very weak requirements.

\begin{lemma}[Composability with $\cdot$]\label{l:concat-composability}
    \begin{subthm}
        \item 
        No order-based semantics is $\cdot$-composable.
        \item 
        No identifier-independent, unbounded, filter-based semantics is $\cdot$-composable.
    \end{subthm}
\end{lemma}

\begin{proof}[Proof of Item 1]
    Consider a semantics $S$ based on some well-partial-order $\preceq$ and the database~$\Dc$ shown below (edge identifiers are not important).

    
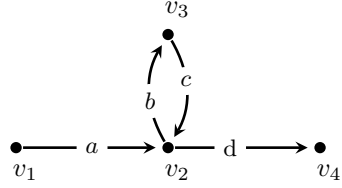
\begin{center}
    \begin{tikzpicture}[bend angle=30]
        \path node[vertex](v1){} ++(-70:1em) node {$v_1$};
        \path (v1) ++(0:2cm) node[vertex](v2){} 
            ++(-70:1em) node {$v_2$};
        \path (v2) ++(0:2cm) node[vertex](v4){} 
            ++(-70:1em) node {$v_4$};
        \path (v2) ++(90:15mm) node[vertex](v3){} 
            ++(70:1em) node {$v_3$};
        \path[edge] (v1) to 
            node[midway,label] {$a$} (v2);
        \path[edge] (v2) to node[label] {d} (v4);
        \path[edge, bend left] (v2) to node[label] {$b$} (v3);
        \path[edge, bend left] (v3) to node[label] {$c$} (v2);
    \end{tikzpicture}
    \end{center}
    
    Consider the RPQ $R = a (bc)^\ast d$. 
    All matches of $R$ between $v_1$ and $v_4$ are labeled by $a (bc)^i d$ for some $i$, and for every $i$ there is a single such match, which we call $w_i$.
    By definition, $S(\Dc, R)$ is the finite and non-empty subset of $\preceq$-minimal elements of $\{w_i \mid i \geq 0\}$. 
    Therefore, there must exist $n, k \geq 0$ such that $w_n \prec w_{n+k+1}$.
    Indeed, choose $m$  such that $w_{m}$ is the longest walk in $S(\Dc, R)$, hence~$w_{m+1}$ does belong to $\matches(\Dc,R)$ but not to $S(\Dc, R)$.
    Hence $S(\Dc, R)$ must contain some~$w_{n}$ with~$n<m$ such that $w_n\prec w_{n+k}$, and finally fix~$k=m-n$. 

    We now define the queries $R_1 = a (bc)^n (b + \varepsilon)$ and $R_2 = (c (bc)^k + \varepsilon) d$. 
    $R_1$ has exactly two matches, the walk $u$ between $v_1$ and $v_2$ labeled $a (bc)^n$ and the walk $u'$ between $v_1$ and $v_3$ labeled $a (bc)^n c$. 
    Since these two walks have different target,
    they both belong to~$S(\Dc,R_1)$.
    Similarly, $R_2$ has exactly two matches, the walk $v$ between $v_2$ and $v_4$ labeled $d$ and the walk $v'$ between $v_3$ and $v_4$ labeled $c(bc)^k d$.
    Similarly, these two walks belong to~$S(\Dc,R_2)$ because they have a different source.

    If $S$ were $\cdot$-composable, then both $uv = w_n$ and $u'v' = w_{n+k+1}$ would be in $S(\Dc, R_1 \cdot R_2)$, which contradicts the fact that $w_n \prec w_{n+k+1}$.
\end{proof}

\begin{proof}[Proof of Item 2]
     Let~$S$ be an identifier-independent unbounded semantics based on filter function~$f$.
    Let~$\Dc,R$ be such that~$S(\Dc,R)$ contains some walk~$w$ satisfying~${\len(w)>0}$.
    Note that~$w$ is a self-consistent walk such that~$f(w)=\top$.
    Let~$\nu$ be any renaming such that~$\nu(\src(w))=\tgt(w)$,~$\nu(\tgt(w))=\src(w)$ and~$\nu(X)\cap X = \emptyset$, where:
    \begin{equation}
    X=(V_\Dc\cup E_\Dc)-\{\src(w),\tgt(w)\}
    \end{equation}
    The walk~$w' = \nu(w)$ belongs to~$\nu(\Dc)$ and since~$S$ is identifier-independent,~$w' \in S(\nu(\Dc),R)$.
    Note that~$w$ and~$w'$ are consistent since they have no edge in common, and both are self-consistent.
    Also note that~$w \cdot w'$ and~$w' \cdot w$ are both well-defined.
    Hence, the characteristic database~$\CharD{w,w'}$ (Def.~\ref{d:characteristic}) exists and contains the walks~$(w \cdot w')^i$ for every~$i\in\N$.
If~$\CharR{w}$ and~$\CharR{w'}$ are the characteristic expressions of~$w$ and~$w'$ respectively (Def.~\ref{d:characteristic}), then~$S(\CharD{w, w'},\CharR{w}) \ni w$ and ${S(\CharD{w, w'}, \CharR{w'}) \ni w'}$.
    Let us assume for the sake of contradiction that~$S$ is $\cdot$-composable.
    Then, $S(\CharD{w,w'},R_i)\ni((w \cdot w')^i)$ for some expression~$R_i$ consisting of concatenations of~$\CharR{w}$ and~$\CharR{w'}$.
    Hence~$f((w \cdot w')^i) = \top$ for every~$i\in\N$, which implies that~$S(\CharD{w,w'},(\CharR{w}\CharR{w'})^*)$ is infinite, a contradiction.
\end{proof}

\begin{lemma}[Composability with $*$]\label{l:star-composability}
\begin{subthm}
    \item \label{l:star-composability:order-based}
    No order-based semantics is $\ast$-composable.
    \item 
    No identifier-independent, unbounded, filter-based semantics is $\ast$-composable.
\end{subthm}
\end{lemma}

\begin{proof}[Proof of Item 1]
    Consider the database~$\Dc$ with one vertex~$v$ and one edge~$e$ labeled by~$a$ from~$v$ to~$v$.
    Then, by definition of an order-based semantics,~$S(\Dc,a)=\{v\rarrow{e}v\}$ and by $\ast$-composability,~$S(\Dc,a^\ast)$ is infinite, a contradiction.
\end{proof}

\begin{proof}[Proof of Item 2]
    For the sake of contradiction, let~$S$ be a semantics that satisfies all four properties of the statement.
    By hypothesis, there is a  database~$\Dc$ and an RPQ~$R$
    such that~$w\in S(\Dc,R)$ and~$\len(w)\geq 1$.
    For every edge~$e$ in~$\Dc$, let $\overline{e}$ be a fresh identifier (i.e., in~$\edges\setminus E$), and
    for every vertex~$v$ in~$\Dc$, let $\overline{v}$ be defined as follows.
    \begin{equation}
        \overline{v} = \begin{cases}
            s & \text{if $v \mathbin= t$}\\
            t & \text{if $v \mathbin= s$}\\
            \text{a fresh identifier in $\vertices\setminus V$ otherwise}\span\omit 
        \end{cases}
    \end{equation}
    We then define the database~$\Dc'=(L_{\Dc'},V_{\Dc'}, E_{\Dc'},\src_{\Dc'},\tgt_{\Dc'},\lbl_{\Dc'})$ as follows:
    \begin{align*}
        L_{\Dc'} ={}& L_\Dc & \src_{\Dc'} \colon e\mapsto{}& \begin{cases}
            \src_{\Dc}(e)&\text{if $e\in E_\Dc$} \\
            \overline{\src_{\Dc}(e)}&\text{if $e\in \overline{E_\Dc}$} \\
        \end{cases} \\
        V_{\Dc'} ={}& {V_\Dc} \cup \overline{V_\Dc} &\omit\span\tgt_{\Dc'}\text{ is defined similarly}\\
        E_{\Dc'} ={}& {E_\Dc} \cup \overline{E_\Dc} &         \lbl_{\Dc'} \colon e\mapsto{}& \begin{cases}
            \lbl_{\Dc}(e)&\text{if $e\in E_\Dc$} \\
            \lbl_{\Dc}(f)&\text{if $e=\overline{f}$ and $f\in E_\Dc$} \\
        \end{cases} 
    \end{align*}
    We lift~$\overline{~\rule{0pt}{1ex}~}$ to walks.
    Note that the database~$\Dc'$ contains the walk~$\overline{w}$ which: starts in~$t$, ends in~$s$, has the same label as~$w$, and is equal to~$w$ up to renaming.
    Therefore $w \cdot \overline{w}$ is a cycle around $s$ with label $uu$ in~$\Dc'$.
    Since~$w$ and~$\overline{w}$ are equal up to renaming
    and~$S$ is identifier-independent, $f(\overline{w})=f(w)=\top$. Hence~$\{w,\overline{w}\}\subseteq S(\Dc',R)$.
    Since~$S$ is $\ast$-composable and $w \cdot \overline{w}$ is a cycle,
    the set~$\{w,\overline{w}\}^*$ is infinite and is a subset of~$S(\Dc',R^*)$, which is a contradiction.
\end{proof}

We conjecture that no ``reasonable'' semantics is composable w.r.t.\@~$\cdot$ nor~$\ast$.
However, it is not easy to define the minimal set of ``reasonable'' properties that implies~$\ast$-non-composability or $\cdot$-non-composability.

To conclude this section, Proposition~\ref{p:compatibilites-specifics} below summarizes properties regarding the compatibilities of some of the semantics defined in the previous sections.

\begin{proposition}\label{p:compatibilites-specifics}
Among the different composability and decomposability properties, only the following hold for the mentioned semantics:
\begin{subthm}
    \item \label{p:acyclic-compatibilites} 
    \acyclicsem{} is only $+$-compatible and $(\cdot, \ast)$-decomposable.
    \item \label{p:trail-compatibilites} \label{p:simple-compatibilites}
    \trailsem{} and \simplesem{} are only $(a,+)$-compatible and $(\cdot,\ast)$-decomposable.
    \item \shortestwalksem, \subwalkminsem, \minmultsem{} are only $a$-compatible and $(\cdot,+,\ast)$-decomposable.
    \item \label{p:binding-trail-compatibilites}
    \bindingtrailsem{} is only $(a,\cdot,+)$-compatible and~$\ast$-decomposable.
\end{subthm}
\end{proposition}

Recall that~$\bindingtrailsem$ is not walk-based (Ex.\,\ref{ex:binding-trail}), hence although it is $\cdot$-composable (item \ref{p:binding-trail-compatibilites}), it is not a counterexample to our conjecture that any ``reasonable'' walk-based semantics cannot be $\cdot$-composable.

\begin{figure}[t]\begin{tikzpicture}
    \renewcommand*{\arraystretch}{0}
    \newcommand{\semname}[1]{#1}
    \newcommand{\hlinit}[1]{\rule[-.2em]{0pt}{1em}\textcolor{bleu}{\textbf{#1}}\small}
    \newlength{\vmh}\setlength{\vmh}{30mm}
    \newlength{\vmv}\setlength{\vmv}{24mm}
    \tikzstyle{sem}=[
         outer sep=2mm, 
         rectangle, 
         rounded corners, 
draw, 
        inner sep=.5ex
    ]
    \tikzstyle{oursem}=[
        sem,    
    ]
    \tikzstyle{incedge}=[draw,
                        ]
    \tikzstyle{classlabel}=[inner sep=2ex]
                    
    \node[sem] (acyclic) at (0,0) {\semname{\hlinit{Ac}yclic}} ;
    \path (acyclic) 
    ++(90:1*\vmv) node[sem] (simple) {\begin{tabular}{@{}l@{}}\hlinit{S}imple\\\hlinit{W}alk or\\\hlinit{C}ycle\end{tabular}}
    ++(90:1*\vmv) node[sem] (trail) {\semname{\hlinit{Tr}ail}};
    
    \path (acyclic) 
    ++(90:.5*\vmv)
    ++(0:1\vmh) 
    node[oursem] (shortest-minimal-set) {\begin{tabular}{@{}l@{}}\hlinit{Sh}ortest\\\hlinit{M}inimal\\\hlinit{S}et\end{tabular}}
    ++(90:\vmv)
    node[oursem] (minimal-multiset) {\begin{tabular}{@{}l@{}}\hlinit{M}inimal\\\hlinit{M}ultiset\end{tabular}}
    ++(90:\vmv) node[oursem] (subwalk-minimal){\begin{tabular}{@{}l@{}}\hlinit{S}ubwalk\\\hlinit{M}inimal\end{tabular}}
    ++(90:\vmv) node[oursem] (binding-trail){\begin{tabular}{@{}l@{}}\hlinit{B}inding\\\hlinit{T}rail\end{tabular}} ;

    \path (acyclic) 
    ++(0:2*\vmh) node[sem] (shortest-walk){\begin{tabular}{@{}l@{}}\hlinit{Sh}ortest walk\end{tabular}}
    ++(1*\vmh,.5*\vmv) node[oursem] (vertex-covering-shortest) {\begin{tabular}{@{}l@{}}\hlinit{Sh}ortest\\\hlinit{V}ertex\\\hlinit{C}over\end{tabular}}
    ++(90:\vmv) node[oursem] (edge-covering-shortest) {\begin{tabular}{@{}l@{}}\hlinit{Sh}ortest\\\hlinit{E}dge\\\hlinit{C}over\end{tabular}};
    
    \path (shortest-walk) ++(90:3*\vmv) node[sem] (atom-covering-shortest) {\begin{tabular}{@{}l@{}}\hlinit{Sh}ortest\\\hlinit{A}tom\\\hlinit{C}over\end{tabular}};

    \newcommand{\inclusion}[3][$\subseteq$]{\path[incedge] (#2) to node[sloped,fill=white] {#1} (#3);}
    
\inclusion{acyclic}{simple}
    \inclusion{simple}{trail}
    \inclusion{trail}{binding-trail}

    \inclusion{shortest-minimal-set}{minimal-multiset}
    \inclusion{minimal-multiset}{subwalk-minimal}
    \inclusion{subwalk-minimal}{binding-trail}
    
    \inclusion{shortest-walk}{vertex-covering-shortest}
    \inclusion{vertex-covering-shortest}{edge-covering-shortest}
    \inclusion[$\supseteq$]{shortest-walk}{shortest-minimal-set}
    \inclusion{acyclic}{shortest-minimal-set}
    
    \inclusion{shortest-walk}{atom-covering-shortest}
    \inclusion[$\supseteq$]{atom-covering-shortest}{binding-trail}

\draw[color=violet, dashed,rounded corners] 
        let  \p1 = (atom-covering-shortest.south),
             \p2 = (subwalk-minimal.north),
             \p3 = (current bounding box.east),
             \p4 = (current bounding box.west) in
        (\x4,\y1/2+\y2/2) -- (\x3+5mm,\y1/2+\y2/2)
        ++(0:-5mm) coordinate (qd) ;
    \path[color=violet] (qd) node[anchor=west]  {\begin{tabular}{r@{}l}$\uparrow$~&Not-walk-based\\[4mm]$\downarrow$&Walk-based\end{tabular}};

\end{tikzpicture}     \caption{Inclusions of main semantics}
    \label{fig:inclusion}
\end{figure}
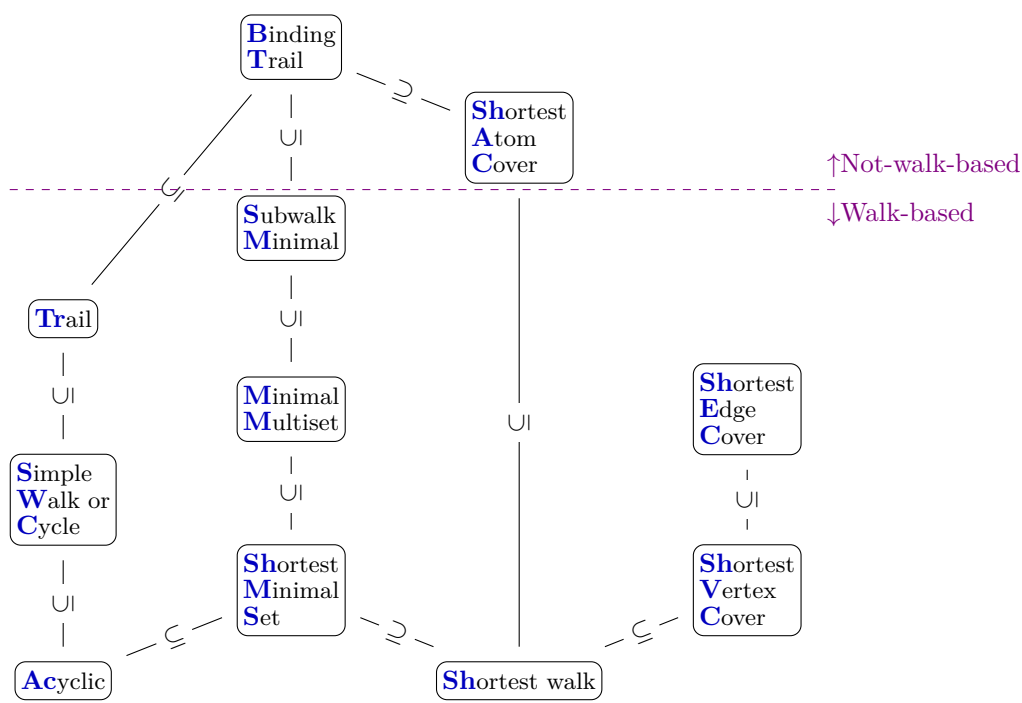

\subsection{Inclusion of semantics}
\label{s:inclusions}

\begin{definition}
    Given two semantics~$S_1,S_2$, we say that~$S_1$ \defemph{is included} in~$S_2$, and denote by~$S_1\subseteq S_2$,
    if it holds~$S_1(\Dc,R)\subseteq S_2(\Dc,R)$ for every database~$\Dc$ and query~$R$.
    We say that~$S_1$ and~$S_2$ are \defemph{incomparable} if neither~$S_1\subseteq S_2$ nor~$S_1\supseteq S_2$ hold.
\end{definition}

Figure~\ref{fig:inclusion} shows the inclusions of the main semantics defined in this document.
All inclusions are strict and the absence of a path means that semantics are incomparable.
Due to the sheer number of them,
we do not include proofs of these claims.
All of them are easy, and some may be found in~\cite{DavidFrancisMarsault2023},~\cite{Marsault2024} and~\cite{Khichane2024}.

\begin{remark}
It is not easy to interpret the fact that two semantics~$S_1,S_2$ satisfy~${S_1\subseteq S_2}$.
It is not always clear whether the results in the difference are mostly meaningful or redundant.
Moreover, there are cases where the difference seems very small
(e.g., \minmultsem{} and \subwalkminsem) but the complexity of query evaluation is substantially different (see
Section~\ref{sec:computational-problems}).
\end{remark}

\subsection{Coverage}
\label{s:coverage}

Recall that a semantics~$S$ returns a finite subset of results ($S(\Dc,R)$) among the possibly infinite set of matches ($\matches(\Dc,R)$).
Some information is bound to be lost in the process.
Then, an interesting question is to characterize the \emph{coverage} of a given semantics $S$, by which we mean \emph{how well $S(\Dc,R)$ represents $\matches(\Dc,R)$}.
In this section, we present a few possible definitions in an attempt to capture this intuition.
On the one hand, coverage can be expressed in terms of vertices (resp.~edges) of the database: all vertices (resp.~edges) in some match must also appear in some result.
On the other hand, coverage may be expressed in terms of subwalks: for every walk $w$ in~$\matches(\Dc,R)$, the set~$S(\Dc,R)$ must contain at least one subwalk of $w$.
Finally, one may be interested in covering the expression instead: if some atom of~$R$ is useful to match any walk, then some walk using it must be in~$S(\Dc,R)$.

\subsubsection{Vertex and edge coverage}
\label{s:vertex-edge-coverage}

\begin{definition}\index{Vertex coverage}
    A walk~$w$ \defemph{covers} a vertex~$v$ if~$v$ occurs in~$w$.
    A set of walks~$W$ \defemph{covers}~$v$ if some walk in~$W$ covers~$v$.
    We say that a semantics~$S$ is \defemph{vertex-covering} if for every~$\Dc,R$
    all vertices covered by~$\matches(\Dc,R)$ are also covered by $S(\Dc,R)$.
\end{definition}

It can be shown that no filter-based semantics is vertex-covering, and any order-based semantics can be made vertex-covering, at the price of an exponential blow-up in the size of the set of results.

\begin{lemma}\label{l:filter=>not-vertex-covering}
    No filter-based semantics is vertex-covering.
\end{lemma}

\begin{proof}
    Consider some semantics $S$ defined by filter function $f$, and let $\Dc$ be the database containing a single vertex $v$ and a single edge $e$ with label $a$ looping around $v$.
    The set $\matches(\Dc, a^\ast)$ contains all walks of the form $w_n = v(ev)^n$ for $n \geq 0$.
    Since $S(\Dc, a^\ast)$ has to be finite, there must exist $k$ such that $w_k \notin S(\Dc, a^\ast)$, i.e. such that $f(w_k) = \bot$. 
    Since this path is the only match in $\Dc$ for query $a^k$, $S(\Dc, a^k)$ must be empty. 
    $S(\Dc, a^k)$ does not cover $v$ even though $\matches(\Dc, a^k)$ does, which shows that $S$ is not vertex-covering.
\end{proof}

\begin{lemma}\label{l:order-based-vertex-covering}
    For any order-based semantics $S$ based on order $\preceq$, there exists a vertex-covering order-based semantics $S'$ based on some order $\preceq'$ that is a refinement of $\preceq$.
\end{lemma}
\begin{proof}
    Let $S$ be a semantics based on the well-partial-order $\preceq$. 
    $S$ might fail to be vertex-covering, because for a given database $\Dc$, vertex $v$ and query $R$ there is no guarantee that $v$ appears in some minimal walk in $\matches(\Dc, R)$.

    Let us define the order $\preceq'$ as: $w \preceq' w'$ if both $w \preceq w'$ and the sets of vertices appearing in $w$ and $w'$ are equal. 
    In other words, we force walks over distinct sets of vertices to be incomparable. 
    One may easily verify that the order-based semantics $S'$ defined by $\preceq'$ is vertex-covering, since for each vertex occurring in $\matches(\Dc, R)$ there must exist at least one minimal element in which this vertex appears (actually at least one for each possible set of vertices appearing in a match containing this vertex).
\end{proof}

By construction of~$\preceq'$ in the proof of Lemma~\ref{l:order-based-vertex-covering}, one may see that
there is an exponential blow-up in the size of the result set in some cases.
For instance, over the $a$-labeled $n$-clique with query $a^\ast$, where the order-based semantics \shortestwalksem{} returns a single result per pair of endpoints, the semantics obtained via this construction would return at least one result per non-empty set of vertices.

We now define a vertex-covering semantics which is neither filter-based nor (strictly speaking) order-based, while avoiding the previous blow-up in the size of results.

\begin{example}[Shortest-vertex-cover Semantics]
    \label{ex:vertex-shortest-covering}
    \semindex{\vertexshortestcovering}{Shortest-vertex-cover semantics} Let $\preceq_{\shortestwalksem}$ be the order used to define $\shortestwalksem$.
    \begin{equation*}
        \vertexshortestcovering(\Dc,R) = \bigcup_{(s,t,v)\in (V_\Dc)^3} \min_{\preceq_{\shortestwalksem}} \{ w\in \matches(\Dc,R) \mid \text{$w$ starts in~$s$, ends in~$t$ and covers~$v$} \}
    \end{equation*}
\end{example}

Note that for each vertex, \vertexshortestcovering{} returns \emph{all} minimal-length matches covering this vertex.
The complexity of query evaluation under \vertexshortestcovering{} is studied in~\cite{Khichane2024} and recalled in Figure~\ref{f:complexity-table}, page \pageref{f:complexity-table}.
The following lemma summarizes some of its other properties.

\begin{lemma} $\vertexshortestcovering$ is 
    \begin{subthm}
        \item \label{lem:vsc-basic} identifier-independent and unbounded;
        \item \label{lem:vsc-not-filter-order} neither filter-based nor order-based;
        \item \label{lem:vsc-comono} co-monotonous, but not monotonous;
        \item \label{lem:vsc-cont} continuous;
        \item \label{lem:vsc-decomp} $a$-compatible and $\theta$-decomposable, but not~$\theta$-composable for any~$\theta\in\{\cdot,+,\ast\}$.
    \end{subthm}
\end{lemma}

The proof of Item 1 is trivial.
The proof of Item \ref{lem:vsc-cont} is similar to the proof that all order-based semantics are continuous in Section~\ref{s:continuity}.

\begin{proof}[Proof of item \ref{lem:vsc-not-filter-order}] 
    Consider database $\Dc$ defined by the graph below:
    
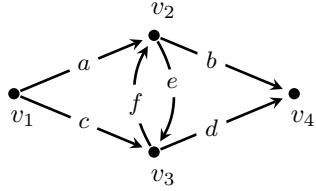
\begin{center}
    \begin{tikzpicture}[bend angle=30]
        \path node[vertex](v1){} ++(-70:1em) node {$v_1$};
        \path (v1) ++(.8cm:2cm) node[vertex](v2){} ++(70:1em) node {$v_2$};
        \path (v1) ++(-.8cm:2cm) node[vertex](v3){} ++(-70:1em) node {$v_3$};
        \path (v2) ++(-.8cm:2cm) node[vertex](v4){} ++(-70:1em) node {$v_4$};
        \path[edge] (v1) to node[midway,label] {$a$} (v2);
        \path[edge] (v2) to node[label] {$b$} (v4);
        \path[edge] (v1) to node[midway,label] {$c$} (v3);
        \path[edge] (v3) to node[label] {$d$} (v4);
        \path[edge, bend left] (v2) to node[label] {$e$} (v3);
        \path[edge, bend left] (v3) to node[label] {$f$} (v2);
    \end{tikzpicture}
    \end{center}
    For simplicity, we can identify each edge with its label.
    Consider the three queries and three walks below.
    \begin{gather}
        R_1 = ab + aefb
        \qquad
        R_2 = cd + aefb
        \qquad
        R_3 = ab + cd + aefb
        \\
        w_1 = v_1 a v_2 b v_4
        \qquad
        w_2 = v_1 c v_3 d v_4
        \qquad
        w_3 = v_1 a v_2 e v_3 f v_2 b v_4
    \end{gather}
    Clearly, $\vertexshortestcovering(\Dc, R_1) = \{w_1, w_3\}$ ($\ast$) and $\vertexshortestcovering(\Dc, R_2) = \{w_2, w_3\}$ ($\dagger$). 
    Walk $w_3$ has to be retained in each set because it covers one more vertex than either $w_1$ or $w_2$ alone.
    Moreover, $\vertexshortestcovering(\Dc, R_3) = \{w_1, w_2\}$ ($\ddagger$) does not contain $w_3$ even though it is also a match for $R_3$, because it is longer than $w_1$ and $w_2$ and covers no additional vertex.

    Assume towards a contradiction that $\vertexshortestcovering$ is defined using some order $\preceq$. 
    By~($\ast$) and~($\dagger$) it must be that $w_1 \not\preceq w_3$ and $w_2 \not\preceq w_3$. 
    However by ($\ddagger$) we must have either $w_1 \preceq w_3$ or $w_2 \preceq w_3$, a contradiction.

    Assume now that $\vertexshortestcovering$ is defined using a filter function $f$. We must have $f(w_3) = \top$ by~($\ast$), but $f(w_3) = \bot$ by ($\ddagger$), also a contradiction.
\end{proof}

\begin{proof}[Proof of item \ref{lem:vsc-comono}]
    To show co-monotonicity, suppose there exist $\Dc \subseteq \Dc'$, $R$ and $w \in \walks(\Dc)$ such that $w \notin \vertexshortestcovering(\Dc, R)$ ($\ast$) but $w \in \vertexshortestcovering(\Dc', R)$~($\dagger$).
    By~($\ast$), then there must exist $w_1, \ldots, w_k \in \vertexshortestcovering(\Dc, R)$ covering all vertices of $w$ and of strictly smaller length than $w$.
    Since $w_1, \ldots, w_k$ are matches for $R$ and belong to $\walks(\Dc')$, they are in $\matches(\Dc', R)$. 
    Therefore, ($\dagger$) cannot hold.

    To show non-monotonicity, consider databases $\Dc$ and $\Dc'$ defined by the graphs below:
    
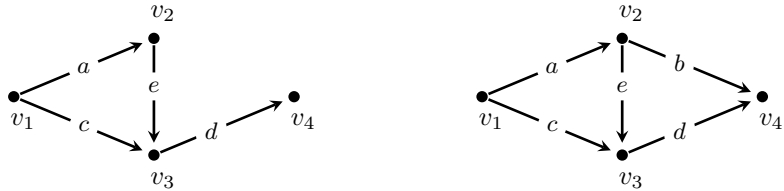
\begin{center}
    \begin{tikzpicture}[bend angle=30]
        \path node[vertex](v1){} ++(-70:1em) node {$v_1$};
        \path (v1) ++(.8cm:2cm) node[vertex](v2){} ++(70:1em) node {$v_2$};
        \path (v1) ++(-.8cm:2cm) node[vertex](v3){} ++(-70:1em) node {$v_3$};
        \path (v2) ++(-.8cm:2cm) node[vertex](v4){} ++(-70:1em) node {$v_4$};
        \path[edge] (v1) to node[midway,label] {$a$} (v2);
\path[edge] (v1) to node[midway,label] {$c$} (v3);
        \path[edge] (v3) to node[label] {$d$} (v4);
        \path[edge] (v2) to node[label] {$e$} (v3);
    \end{tikzpicture}
    \hfil
    \begin{tikzpicture}[bend angle=30]
        \path node[vertex](v1){} ++(-70:1em) node {$v_1$};
        \path (v1) ++(.8cm:2cm) node[vertex](v2){} ++(70:1em) node {$v_2$};
        \path (v1) ++(-.8cm:2cm) node[vertex](v3){} ++(-70:1em) node {$v_3$};
        \path (v2) ++(-.8cm:2cm) node[vertex](v4){} ++(-70:1em) node {$v_4$};
        \path[edge] (v1) to node[midway,label] {$a$} (v2);
        \path[edge] (v2) to node[label] {$b$} (v4);
        \path[edge] (v1) to node[midway,label] {$c$} (v3);
        \path[edge] (v3) to node[label] {$d$} (v4);
        \path[edge] (v2) to node[label] {$e$} (v3);
    \end{tikzpicture}
    \end{center}
    and query $R = ab + cd + aed$. Clearly, $\Dc \subseteq \Dc'$, and the $aed$-labeled walk between $v_1$ and $v_4$ belongs to $\vertexshortestcovering(\Dc, R)$ but not to $\vertexshortestcovering(\Dc', R)$.
\end{proof}

\begin{proof}[Proof of item \ref{lem:vsc-decomp}]
    First, let us show that~$\vertexshortestcovering$ is $\ast$-decomposable.
    Let~$\Dc$ be a database,~$R$ a query and~$w$ be any walk in~$\vertexshortestcovering(\Dc,R^*)$. By definition of $\vertexshortestcovering$, there exists a vertex $v$ such that $w$ is the shortest among all walks that cover~$v$ and match $R$.
    Since~$w \in \matches(\Dc,R^*)$, there is a factorization~$w = w_1 \cdots w_k$ such that~$w_i\in\matches(\Dc,R)$ for every~$i\in\{1,\ldots,k\}$.
    For the sake of contradiction, let us assume that~$w_i\notin\vertexshortestcovering(\Dc,R)$ for some~$i$.
    We consider two cases.
    \begin{enumerate}
        \item If there exists a factor $w_j$ with $j \neq i$ visiting $v$, since $w_i\notin\vertexshortestcovering(\Dc,R)$ there must exist a walk $w'_i$ strictly shorter than $w_i$ and with the same endpoints in $\vertexshortestcovering(\Dc,R)$. 
        Then, let~$w'$ be the walk resulting from substituting $w'_i$ for $w_i$ in the factorization of $w$.
        The walk~$w'$ is a strictly shorter than~$w'$, has the same endpoints, has a label in $R^\ast$ and covers $v$.
        It is a contradiction to~$w \in \vertexshortestcovering(\Dc,R^*)$.
        \item Otherwise, since $w_i\notin\vertexshortestcovering(\Dc,R)$ there must exist some walk $w'_i$ in $\vertexshortestcovering(\Dc,R)$ with the same endpoints as $w_i$ covering $v$. 
        We conclude as before by substituting $w'_i$ for $w_i$ in $w$.
    \end{enumerate}
    
    Showing that \vertexshortestcovering is $\cdot$-decomposable is similar 
    but easier, and showing that it is $+$-decomposable is trivial.
    Non-$\ast$-composability is shown using the same argument as Item~\ref{l:star-composability:order-based} of Lemma~\ref{l:star-composability}.
    
    One may show non-$+$-composability and non-$\cdot$-composability using database~$\Dc'$ from item~3.
    For non-$+$-composability, consider queries~$R_1=ab+cd$ and~$R_2=aed$. The walk $w$ between $v_1$ and $v_4$ labeled $aed$ is in $S(\Dc', R_2)$ but not in $S(\Dc', R_1 + R_2)$.
    For non-$+$-composability, consider queries~$R_3=ae+c$ and~$R_4=b+d$. The walk $w$ between $v_1$ and $v_4$ labeled $aed$ is in $S(\Dc', R_1) \cdot S(\Dc', R_2)$ but not in $S(\Dc', R_1 \cdot R_2)$.
\end{proof}

Note that replacing $\preceq_{\shortestwalksem}$ by any suitable order in Example~\ref{ex:vertex-shortest-covering} yields a valid RPQ semantics.
This opens up a whole new way of defining semantics that we leave for future work.

\subparagraph{Edge Coverage.}
The situation with respect to edge coverage is much like the one for vertex coverage.
It may be defined similarly, and no semantics defined in this document is edge-covering.
A shortest covering semantics, the shortest-edge-cover semantics~\edgeshortestcovering, may be defined by adapting Example~\ref{ex:vertex-shortest-covering}.

\subsubsection{Subwalk guarantee}
\label{s:subwalk-guarantee}

\begin{definition}\index{Subwalk guarantee}
    A semantics~$S$ possesses the \defemph{subwalk guarantee}
    if for every database~$\Dc$, query~$R$, and walk~$w$
    such that~$w\in\matches(\Dc,R)$ there exists a subwalk of~$w$ that belongs to~$S(\Dc,R)$.
\end{definition}

This property was introduced in Lemma~13 of \cite{DavidFrancisMarsault2023},
which states that $\bindingtrailsem$ possesses it.
This led to the definition of $\subwalkminsem$, which possesses it by definition.
Other semantics discussed in this document do not possess it.

\begin{lemma}
    \begin{subthm}
    \item 
    $\subwalkminsem$ and  $\bindingtrailsem$ possess the subwalk coverage guarantee.
    \item $\vertexshortestcovering$, $\shortminsetsem$, $\minmultsem$ and~$\shortestwalksem$ do not possess the subwalk coverage guarantee.
    \item No filter-based semantics possesses the subwalk guarantee.
    \end{subthm}
\end{lemma}

\begin{proof}[Proof for item 1] 
    $\bindingtrailsem$ possess the subwalk coverage guarantee from \cite[Lemma 13]{DavidFrancisMarsault2023}.
    $\subwalkminsem$ possess it by definition.
\end{proof}

\begin{proof}[Proof for item 2] 
    For $\shortminsetsem$, $\minmultsem$ and~$\shortestwalksem$ the results comes from the fact that they are strictly included in~$\subwalkminsem$.
    Hence on the instances where the inclusion are counterexamples of the subwalk coverage guarantee.

    For $\vertexshortestcovering$, consider the database~$\Dc$ consisting of one $a$-cycle and two $b$-cycles given below.
    \begin{gather}
      \label{eq:acycle}v_1\xrightarrow{a}v_2\xrightarrow{a}   v_3\xrightarrow{a}v_1 \\
      v_1\xrightarrow{b}v_2\xrightarrow{b}v_1 \\
      v_1\xrightarrow{b}v_3\xrightarrow{b}v_1
    \end{gather}
    For~$R=aaa+bb$, it may be verified that no subwalk of the walk in \eqref{eq:acycle} is~$\vertexshortestcovering(\Dc,R)$.
\end{proof}
\begin{proof}[Proof sketch for item 3] 
    A filter-based semantics has at least one walk~$w$ of positive length that it filters out.
    Consider the characteristic database~$\Dc$ and
    characteristic regular expression~$R$ of~$w$
    (Def.~\ref{d:characteristic}).
    Then~$S(\Dc,R)$ is empty while~$w$ belongs to~$\matches(\Dc,R)$.
\end{proof}

\subsubsection{Atom coverage}

Another, orthogonal type of coverage is worth mentioning.
It covers the regular expression instead of the match set.
Intuitively, we say a semantics $S$ is atom-covering if for every $\Dc$ and $R$, and for every letter occurrence in $R$, $S(\Dc, R)$ contains at least one walk matched using this letter occurrence if any such walk exists.
We give below a formal definition.

\begin{definition}\label{d:atom-covering}\index{Atom coverage}
    Let $R$ be a regular expression over $\Sigma$. 
    We denote by~$k$ the number of atoms in~$R$, and by $a_i$ the $i$-th atom of~$R$, for~$i\in\{1,\ldots,k\}$.
    \begin{subthm}
        \item \label{d:regexp-linearisation}
        A \defemph{linearisation} of $R$ is a copy $R'$ of $R$ in which each atom occurrence in $\Sigma$ is replaced by a different symbol in a new alphabet $\Gamma$. 
        We denote by~$\alpha$ the corresponding bijection~$\{1,\ldots,k\}\to\Gamma$ and by~$\beta$ the function $\Gamma\to\Sigma$ defined by: $\forall i\in\{1,\ldots,k\}$
        $\beta(\alpha(i))=a_i$.
        \item 
        A word~$u$ \defemph{covers the $i$-th atom} of~$R$ if there exists some word~$u'\in L(R')$ in which the letter $\alpha(i)$ occurs, and such that $u=\beta(u')$, for some linearisation~$R'$ of~$R$.
        A walk~$w$ in~$\Dc$ covers the $i$-th atom of~$R$ if its label does.

        Note that this notion is independent of the chosen linearisation.
        
        \item A semantics~$S$ is \defemph{atom covering} if
        for every database~$\Dc$, regular expression~$R$ and position~$i$ in~$\{1,\ldots,k\}$
        such that some walk in $\matches(\Dc,R)$ covers the~$i$-th atom of~$R$,
        there is a walk covering the~$i$-th atom of~$R$ in $S(\Dc,R)$.
    \end{subthm}
\end{definition}

This definition expresses the idea that when designing a query with multiple alternatives, each alternative which yields at least a match should be reflected in the result set.
It may be verified that no walk-based semantics can be atom-covering.
As was the case for vertex and edge coverage, it is quite easy to define atom-covering semantics (although they are not walk-based in this case), such as the one given next.

\begin{example}[Shortest atom-cover $\atomshortestcoverer$]
\semindex{\atomshortestcoverer}{Shortest-atom-cover semantics}
Let~$\Dc$ be a database and~$R$ be a regular expression.
Let~$\{1,\ldots,k\}$ be the set of positions of atoms in~$R$. Let $\preceq_{\shortestwalksem}$ be the \emph{shorter than} order (on which $\shortestwalksem$ is based). Then:
\begin{equation*}
        \atomshortestcoverer(\Dc,R) = \bigcup_{s,t\in V_\Dc}   
        \bigcup_{i=1}^{k}
        \min_{\preceq_{\shortestwalksem}} \left\{ w\in \matches(\Dc,R) ~\middle| 
        \begin{array}{l}\text{$w$ starts in~$s$ and ends in~$t$}
         \\\text{$w$ covers the $i$-th atom of~$R$}
        \end{array}\right\}
    \end{equation*}
\end{example}

It is routine to check that the definition of \atomshortestcoverer{} implies that it is atom-covering.
Since \bindingtrailsem{} includes \atomshortestcoverer{} semantics (See Section~\ref{s:inclusions}), it is also atom-covering.

\begin{proposition}
    \atomshortestcoverer{} and \bindingtrailsem{} are atom-covering.
\end{proposition}

\subsection{Summary: what is a reasonable semantics?}
\label{s:reasonable-semantics}

Although many of the properties described in Sec.~\thesection{} are desirable, we saw that some of them are incompatible.
We list below the properties that we think any reasonable semantics should satisfy.
As we will see, there are caveats in a lot of cases.

\begin{description}
    \item
    [Identifier independence \normalfont(Def.~\ref{d:identifier-independent})]
    In database systems, identifiers are not always accessible to the user, nor guaranteed to remain constant over time.
    Therefore, it seems strange to have a semantics rely on identifiers.
    Note however that identifier-dependent semantics might still be useful to model cases where the database engine is allowed to make arbitrary choices, like the \gql{ANY} \gql{SHORTEST} keyphrase in GQL or the \shortlexsem{} semantics (Ex.~\ref{ex:shortlex}).
    
    \item[Label-independence \normalfont(Def.~\ref{d:label-independent})] 
    Semantics are defined before (hence independently of) database schemas.
    Since labels are abstractions of relation names in schema,
    it seems strange to attach a meaning to them.
    On the other hand, one could use labels to model other aspects, such as the weight in a semantics
    like the one described by the \gql{ANY} \gql{CHEAPEST} keyphrase proposal in GQL (see Ex.~\ref{ex:cheapest}).
    
    \item[Unboundedness \normalfont(Def.~\ref{d:unbounded-semantics})] Bounded semantics only yield walks whose length is bounded by some constant \emph{regardless of the query or the database}, which seems to defeat the whole purpose of using RPQs. Note that we are referring here to a property of the semantics itself, not to possible extensions of the query language allowing users to limit the length of returned walks (which could be useful in certain scenarios).
    
    \item[Co-monotony \normalfont(Def.~\ref{d:comonotonous})] A co-monotonous semantics ensures that, for a fixed query, whenever extending a database graph with new vertices or edges, the new results should always use some of the new vertices or edges.
    On the contrary, under a non-co-monotonous semantics, whether a match belongs to the result set or not depends on information that seems irrelevant.
    
    \item[Decomposability \normalfont(Def.~\ref{d:decomposability})]
    Similarly to co-monotony, semantics that are not decomposable 
    sometimes have new results that materialize out of thin air. In particular, for such semantics, results cannot be computed bottom up from the expression.
    
    \item[Compatibility with atoms \normalfont(Def.~\ref{d:atom-compatibility})]
    Using a semantics that is not compatible with atoms means that some edges might be missing from replies to atomic queries.  If the semantics is moreover decomposable, these missing edges will never appear in the results of any query.
    Note that~$\acyclicsem$, one of the semantics available in GQL, is in fact not compatible with atoms, due to possible self-loops in the graph. 
    This probably means that \acyclicsem{} assumes the graph to have no self-loops.
    
    \item[Continuity \normalfont(Def.~\ref{d:continuous})] 
    A (walk-based) semantics that is not continuous means that the semantics behaves differently for finite match sets and infinite match sets.
    Note that a semantics that is not expression-independent (hence also not walk-based, see Lemma~\ref{l:walk-based=>expression-independent}), such as binding-trail semantics (Example~\ref{ex:binding-trail}), may be  noncontinuous but still be considered meaningful.
    
    \item[Complexity]
    The definition of filter-based (Def.~\ref{d:filter-based}) and order-based (Def.~\ref{d:order-based}) semantics makes no requirement about the complexity class of the filter function or of the order.
    Reasonable semantics should be based on filter functions or orders that are tractable, otherwise the complexity classes of usual query evaluations are bound to be untractable (see Sec.~\ref{sec:computational-problems}).
\end{description}

\begin{figure}[ht]
    \newcolumntype{L}{>{\rule{0pt}{2.6ex}\raggedright\arraybackslash}X<{\rule[-1.4ex]{0pt}{2.5ex}}}
    \newcolumntype{C}{>{\rule{0pt}{2.6ex}\centering\arraybackslash}X<{\rule[-1.4ex]{0pt}{2.5ex}}}
    \centering \begin{tabularx}{\linewidth}{>{\raggedright}p{46mm}cccC}
    \hline
    \rowcolor{tabularrowtitle}
    \multicolumn{1}{l}{Semantics} & Existence & Memb. & Extensibility & Enumeration
    \\
    \hline
    \trailsem{} (Trail) \cite{MartensNiewerthPopp2023}    & \np-c        &   \nl        & \np-c & Preproc. is \np-c
    \\
    \rowcolor{tabularrowcolor} 
    \acyclicsem{} (Acyclic) \cite{BaganBonifatiGroz2020} & \np-c
      & \nl & \np-c & Preproc. is \np-c 
    \\
    \shortestwalksem{} (Shortest walk)& \nl-c & \ptime & \ptime& \polydelay 
    \\
    \rowcolor{tabularrowcolor} 
    \bindingtrailsem{} (Binding trail) \cite{DavidFrancisMarsault2023} & \nl-c                  & \np-c & \np-c& Open/\polydelay\footnote{Enumeration under Binding-trail semantics is \polydelay{} under bag semantics, that is if the same walk is output once for each way it is matched.}
    \\
    \subwalkminsem{} (Subwalk minimal) \cite{Khichane2024} & \nl-c & \ptime &  \np-c & Open 
    \\
    \rowcolor{tabularrowcolor} 
    \minmultsem{}  (Minimal Multiset) \cite{Marsault2024} & \nl-c &
    \conp-c& \conp-h &
    Not \polydelay{} unless $\ptime=\np$
    \\
    \vertexshortestcovering{}  (Short.\@ Vertex Cover) {\cite{Khichane2024}} & \nl-c &
    \ptime & \ptime &
    \polydelay
    \\
    \hline
    \end{tabularx}
    \caption{Summary of computational problem (data) complexity for the main semantics}
    \label{f:complexity-table}
\end{figure}

\section{Computational problems}
\label{sec:computational-problems}

This document is \textbf{not} about computational problems. 
In Section~\thesection{}, we gather computational problems and known complexity results for them from other sources.  In particular, Figure~\ref{f:complexity-table} sums up known results about
the most interesting semantics mentioned in this document.

\begin{problem}{Existence under $S$}
    \item[Data:] A database~$\Dc$ and two vertices $s$ and $t$ from~$\Dc$.
    \item[Query:] A regular expression~$R$.
    \item[Output:] Is there a walk $w$ from~$s$ to~$t$ in $S(\Dc,R)$ ?
\end{problem}

\begin{theorem}
    Existence under $S$ is in $\nl$ whenever~$S$ is order-based.
\end{theorem}

As recalled in Figure~\ref{f:complexity-table}, it is well-known that \problemfont{Existence}
is \np-complete under the semantics \trailsem{} and \acyclicsem.
It is quite easy to show that it is also \np-complete under the variant \simplesem{} (Ex.~\ref{ex:simple})
or \twoacyclicsem{} (Ex.~\ref{ex:2-ac}).
Actually, we conjecture that \problemfont{Existence} is \np-hard for any filter-based semantics:

\begin{conjecture}\label{c:filter-based=>existence-NP-hard}
    \problemfont{Existence under $S$} is \np-hard in data-complexity for every ``reasonable'' filter-based semantics.
\end{conjecture}

We added the restriction \emph{reasonable} in Conjecture~\ref{c:filter-based=>existence-NP-hard} because the problem is likely to be void otherwise.
Indeed, definition of filter-based semantics (Def.~\ref{d:filter-based}) does not impose restrictions on the complexity of the filter function.
Moreover, the proof is likely to require the semantics to satisfy other properties, such as the one provided in Sec.~\ref{s:reasonable-semantics}.

\begin{problem}{Membership under $S$}
    \item[Data:] A database~$\Dc$ and a walk $w$ in~$\Dc$.
    \item[Query:] A regular expression~$R$.
    \item[Output:] Does $w$ belong to~$S(\Dc,R)$ ?
\end{problem}

Once again, the definition of filter-based semantics does not impose any
complexity restriction on the filter function.
Hence, \problemfont{Membership} under a filter-based semantics is computationally equivalent to the filter function.
The situation is worse for order-based semantics: Membership under $S_{\leq}$ may be harder than the complexity of $\leq$. It is for instance \conp-hard under
\minmultsem{} \cite{Marsault2024}, while the order itself is in \nl.
It is the worst-case scenario since \problemfont{Membership} under a reasonable order-based semantics is always in \conp, as stated below.

\begin{proposition}
    Let~$S_\preceq$ is a semantics based on an order~$\preceq$ in \ptime.
    If there is a polynomial function~$p$ such that for every walks~$w'\preceq w$, it holds~$|w'|< p(|w|)$, then \problemfont{Membership} under $S$ is in \conp.
\end{proposition}
\begin{proof}[Sketch of proof] 
    Non-membership of~$w$ amounts to finding a match~$w'$ that satisfies~${w'<w}$.
    Checking that a given~$w'$ satisfies both conditions may be done in polynomial time by hypotheses.
    On the other hand, any~$w'$ satifying~$w'<w$ must have a length that is polynomial in the length of~$w$, hence~$w'$ is necessarily short enough to serve as a certificate for non-membership of~$w$.
\end{proof}

\begin{problem}{Enumeration under $S$}
    \item[Data:] A database~$\Dc$ and a walk $w$ in~$\walks(\Dc)$.
    \item[Query:] A regular expression~$R$.
    \item[Output:] All walks $w$ in~$S(\Dc,R)$ going from~$s$ to~$t$.
\end{problem}

The lower bound for the preprocessing time of \problemfont{Enumeration} is at least as high as the lower bound of \problemfont{Existence}, since one could use an algorithm for the former to solve the latter.
This is the argument used to show that \problemfont{Enumeration} under usual filter-based semantics (\trailsem, \acyclicsem, \simplesem, etc.) is intractable.
The case of~\minmultsem{} is remarkable in the sense that \problemfont{Enumeration} is intractable, but due to the delay and not the preprocessing \cite{Marsault2024}.

We introduce below the less common problem of \problemfont{Extensibility}.
It allows to use the well-known technique of the \emph{flash-light search}: \problemfont{Enumeration} is tractable, whenever both  \problemfont{Extensibility} and \problemfont{Membership} are.

\begin{problem}{Extensibility under $S$}
    \item[Data:] A database~$\Dc$, a walk $w$ in~$\walks(\Dc)$ and a vertex~$t$ in~$\Dc$.
    \item[Query:] A regular expression~$R$.
    \item[Output:] Is there a walk $w'$ such that $w\cdot w'$ belong to~$S(\Dc,R)$ and ends in~$t$? 
\end{problem}

\problemfont{Extensibility} is at least as hard as \problemfont{Existence} since \problemfont{Existence} amounts to providing empty walk as input to \problemfont{Extensibility}.
Although it is not formally true, the definition of a semantics usually implies that \problemfont{Extensibility} is usually at least as hard as \problemfont{Membership}.

\begin{proposition}[Application of {\protect\cite[Prop.~3.5]{Strozecki2021}}]
    \label{p:extensibility+Membership=>Enumeration}
    If \problemfont{Extensibility} and \problemfont{Membership} under a semantics~$S$ are both in \ptime{} in data complexity, 
    then \problemfont{Enumeration} under~$S$ is \polydelay{} in data complexity (and memoryless, see \cite{Strozecki2021}).
\end{proposition}
\begin{proof}[Sketch of proof] 
Let A database~$\Dc$, a walk $w$ in~$\walks(\Dc)$ and a vertex~$t$ in~$\Dc$
be the input of the problem. 
For every vertex~$s$ of the database we call the function \textsc{Forward} below, with the trivial walk~$(s)$ as argument.
    \begin{vmalgorithm}
        \Function{Forward}{Walk $w$}
            \If{$\tgt(w)= t$\textbf{ and }\problemfont{Membership}($\Dc, R, w$) \text{ returns }\True}
                \State \textbf{Output }$w$
            \EndIf
            \ForAll{edge $e$ going out from $\tgt(w)$}
                \State $w' \gets w \cdot (\src(e), e, \tgt(e)) $
                \If{\problemfont{Extensibility}($\Dc, R, w'$, $t$) returns \True}
                    \State \Call{Forward}{w'}
                \EndIf
            \EndFor
            \State \Return $\bot$
        \EndFunction   
    \end{vmalgorithm}
    Note that the number of recursive calls to \textsc{Forward} on the stack is bounded by the length of the next walk to output.
\end{proof}

\begin{remark}
   There are variants of \problemfont{Existence} and \problemfont{Enumeration}
    in which endpoints as not provided as input.
    Similarly, sometimes the target~$t$ of \emph{Extensibility} is not provided.
    These changes usually have no impact on the complexity of those problems.
\end{remark}

\section{Conclusion and future work}

Only a few RPQ semantics have been considered in industry, and their principled study was generally performed a posteriori.
In this work, we took the opposite approach: first define the properties a RPQ semantics could have, in order to design better RPQ semantics in the future.
We distinguished two simple classes of semantics: filter-based semantics which select results using a per-walk criterion, and order-based semantics which only retain minimal matches with respect to some ordering relation.
We proposed a list of possible properties, some of which may be considered requirements for all RPQ semantics, while others turned out to be incompatible.
This process yielded several promising semantics, such as subwalk-minimal and shortest-vertex-cover semantics, which probably deserve further investigations.

Our work is by no means exhaustive. It has several limitations, which constitute interesting leads for future work.
First, some of our definitions may be too restrictive. 
For instance, the class of \emph{filter-based} semantics is tailored to capture several semantics used in practice, but our definition appears to leave little room for novel designs.
Another example is compatibility with rational operators. Most of them are verified by very few reasonable semantics, and weaker notions of compatibility might be worth considering.
For instance, a weaker version of $\ast$-composability could be that $S(\Dc,R)$ is always included in $S(\Dc,R^*)$.

Second, we limited the scope of our study to walk-based semantics.
A semantics is a way to restrict the potentially infinite set of matches to  a finite and hopefully representative set of results for the user.  From this perspective, walk-based semantics must make this selection using a very limited amount of information.
This prevents us from considering semantics whose results depend on the actual syntax of the regular expression, such as the ones in \cite{ArenasConcaPerez2012,DavidFrancisMarsault2023}.
Future work could investigate semantics that make use of more information, such as walk labels, the entire database, the syntax of the query, or how walks are matched by the query.

Third, both our data model and the query language we consider are rather simplistic.
Adding features such as data or node labels should not affect our observations regarding walk-based semantics. 
However, these changes increase the design space when considering semantics that are not walk-based (as discussed above).
Note that the GQL data model~\cite{DeutschEtAl2022,GQL-ISO} allows edges (and nodes) to bear multiple labels.
This raises the question of walk multiplicity in the result, a dimension we ignored entirely. 
It would also be interesting to study \emph{dynamical} semantics:
enrich the syntax to give users some control over the semantics, similar to \gql{CHEAPEST} in GQL.

\clearpage
\section{Bibliography}
\nocite{SQL-ISO}
\subsection{Academic publications}
\printbibliography[notkeyword=language,heading=none]
\subsection{Language documentations and specifications}
\printbibliography[keyword=language,heading=none]

\clearpage
\phantomsection
\label{s:index}
\printindex
\end{document}